\documentclass[11pt]{article}

\usepackage[dvipsnames]{xcolor}
\usepackage{framed}
\definecolor{shadecolor}{rgb}{0.9,0.9,0.9}
\usepackage{mathtools}
\usepackage{amsmath}
\usepackage[shortlabels]{enumitem}

\usepackage{graphicx,epic,eepic,epsfig,amsmath,latexsym,amssymb,verbatim,color}
 
\usepackage{amsfonts}       
\usepackage{nicefrac}       

\usepackage{amsmath}
\usepackage{bbm}

\usepackage{float}
\usepackage{tikz}
\usetikzlibrary{chains}
\usetikzlibrary{fit}
\usetikzlibrary{arrows} 

\usepackage{epsfig}
\usetikzlibrary{shapes.symbols,patterns} 
\usepackage{pgfplots}

\usepackage[strict]{changepage}
\usepackage{hyperref}
\hypersetup{colorlinks=true,citecolor=blue,linkcolor=blue,filecolor=blue,urlcolor=blue,breaklinks=true}

\usepackage[marginal]{footmisc}
\usepackage{url}
\usepackage{theorem}

\theoremstyle{plain}
\newtheorem{proposition}{Proposition}
\newtheorem{lemma}[proposition]{Lemma}
\newtheorem{theorem}[proposition]{Theorem}
\newtheorem{corollary}[proposition]{Corollary}

\theorembodyfont{\upshape}
\newtheorem{definition}{Definition}

\def\squareforqed{\hbox{\rlap{$\sqcap$}$\sqcup$}}
\def\qed{\ifmmode\squareforqed\else{\unskip\nobreak\hfil
\penalty50\hskip1em\null\nobreak\hfil\squareforqed
\parfillskip=0pt\finalhyphendemerits=0\endgraf}\fi}
\def\endenv{\ifmmode\;\else{\unskip\nobreak\hfil
\penalty50\hskip1em\null\nobreak\hfil\;
\parfillskip=0pt\finalhyphendemerits=0\endgraf}\fi}
  
\newenvironment{proof}{\noindent \textbf{{Proof} }}{\hfill $\blacksquare$}

\newcounter{remark}
\newenvironment{remark}[1][]{\refstepcounter{remark}\par\medskip\noindent%
\textbf{Remark~\theremark #1} }{\medskip}

\newcounter{example}

\mathchardef\ordinarycolon\mathcode`\:
\mathcode`\:=\string"8000
\def\vcentcolon{\mathrel{\mathop\ordinarycolon}}
\begingroup \catcode`\:=\active
  \lowercase{\endgroup
  \let :\vcentcolon
  }

\usepackage{cleveref}
\usepackage{graphicx}
\usepackage{xcolor}

\RequirePackage[framemethod=default]{mdframed}
\newmdenv[skipabove=7pt,
skipbelow=7pt,
backgroundcolor=darkblue!15,
innerleftmargin=5pt,
innerrightmargin=5pt,
innertopmargin=5pt,
leftmargin=0cm,
rummaging=0cm,
innerbottommargin=5pt,
linewidth=1pt]{tBox}

\newmdenv[skipabove=7pt,
skipbelow=7pt,
backgroundcolor=red!15,
innerleftmargin=5pt,
innerrightmargin=5pt,
innertopmargin=5pt,
leftmargin=0cm,
rightmargin=0cm,
innerbottommargin=5pt,
linewidth=1pt]{rBox}

\newmdenv[skipabove=7pt,
skipbelow=7pt,
backgroundcolor=blue2!25,
innerleftmargin=5pt,
innerrightmargin=5pt,
innertopmargin=5pt,
leftmargin=0cm,
rightmargin=0cm,
innerbottommargin=5pt,
linewidth=1pt]{dBox}
\newmdenv[skipabove=7pt,
skipbelow=7pt,
backgroundcolor=darkkblue!15,
innerleftmargin=5pt,
innerrightmargin=5pt,
innertopmargin=5pt,
leftmargin=0cm,
rightmargin=0cm,
innerbottommargin=5pt,
linewidth=1pt]{sBox}
\definecolor{darkblue}{RGB}{0,76,156}
\definecolor{darkkblue}{RGB}{0,0,153}
\definecolor{blue2}{RGB}{102,178,255}
\definecolor{darkred}{RGB}{195,0,0}

\newcommand{\nc}{\newcommand}
\nc{\rnc}{\renewcommand}
\nc{\lbar}[1]{\overline{#1}}
\nc{\bra}[1]{\langle#1|}
\nc{\ket}[1]{|#1\rangle}
\nc{\ketbra}[2]{|#1\rangle\!\langle#2|}
\nc{\braket}[2]{\langle#1|#2\rangle}

\nc{\proj}[1]{| #1\rangle\!\langle #1 |}
\nc{\avg}[1]{\langle#1\rangle}
\nc{\rank}{\operatorname{Rank}}
\nc{\smfrac}[2]{\mbox{$\frac{#1}{#2}$}}
\nc{\tr}{\operatorname{Tr}}
\nc{\ox}{\otimes}
\nc{\dg}{\dagger}
\nc{\dn}{\downarrow}
\nc{\cA}{{\cal A}}
\nc{\cB}{{\cal B}}
\nc{\cC}{{\cal C}}
\nc{\cD}{{\cal D}}
\nc{\cE}{{\cal E}}
\nc{\cF}{{\cal F}}
\nc{\cG}{{\cal G}}
\nc{\cH}{{\cal H}}
\nc{\cI}{{\cal I}}
\nc{\cJ}{{\cal J}}
\nc{\cK}{{\cal K}}
\nc{\cL}{{\cal L}}
\nc{\cM}{{\cal M}}
\nc{\cN}{{\cal N}}
\nc{\cO}{{\cal O}}
\nc{\cP}{{\cal P}}
\nc{\cQ}{{\cal Q}}
\nc{\cR}{{\cal R}}
\nc{\cS}{{\cal S}}
\nc{\cT}{{\cal T}}
\nc{\cU}{{\cal U}}
\nc{\cV}{{\cal V}}
\nc{\cX}{{\cal X}}
\nc{\cY}{{\cal Y}}
\nc{\cZ}{{\cal Z}}
\nc{\cW}{{\cal W}}
\nc{\ip}{{\mathit{p}}}
\nc{\iq}{{\mathit{q}}}
\nc{\csupp}{{\operatorname{csupp}}}
\nc{\qsupp}{{\operatorname{qsupp}}}
\nc{\var}{{\operatorname{var}}}
\nc{\rar}{\rightarrow}
\nc{\lrar}{\longrightarrow}
\nc{\polylog}{{\operatorname{polylog}}}
\nc{\wt}{{\operatorname{wt}}}
\nc{\av}[1]{{\left\langle {#1} \right\rangle}}
\nc{\supp}{{\operatorname{supp}}}

\nc{\argmin}{{\operatorname{argmin}}}

\def\x{\xi}

\nc{\RR}{{{\mathbb R}}}
\nc{\CC}{{{\mathbb C}}}
\nc{\FF}{{{\mathbb F}}}
\nc{\NN}{{{\mathbb N}}}
\nc{\ZZ}{{{\mathbb Z}}}
\nc{\PP}{{{\mathbb P}}}
\nc{\QQ}{{{\mathbb Q}}}
\nc{\UU}{{{\mathbb U}}}
\nc{\EE}{{{\mathbb E}}}
\nc{\id}{{\operatorname{id}}}

\nc{\CHSH}{{\operatorname{CHSH}}}

\nc{\be}{\begin{equation}}
\nc{\ee}{{\end{equation}}}
\nc{\bea}{\begin{eqnarray}}
\nc{\eea}{\end{eqnarray}}
\nc{\<}{\langle}
\rnc{\>}{\rangle}
\nc{\rU}{\mbox{U}}

\nc{\ob}[1]{#1}

\nc{\SEP}{{\text{\rm SEP}}}
\nc{\NS}{{\text{\rm NS}}}
\nc{\LOCC}{{\text{\rm LOCC}}}
\nc{\PPT}{{\text{\rm PPT}}}
\nc{\EXT}{{\text{\rm EXT}}}
\nc{\Sym}{{\operatorname{Sym}}}

\nc{\ERLO}{{E_{\text{r,LO}}}}
\nc{\ERLOCC}{{E_{\text{r,LOCC}}}}
\nc{\ERPPT}{{E_{\text{r,PPT}}}}
\nc{\ERLOCCinfty}{{E^{\infty}_{\text{r,LOCC}}}}
\nc{\Aram}{{\operatorname{\sf A}}}

\newcommand{\CPTP}{\text{\rm CPTP}}
\DeclareMathOperator{\Tr}{Tr}

\usepackage{tikz}
\usepackage{hyperref}
\hypersetup{colorlinks=true,citecolor=blue,linkcolor=blue,filecolor=blue,urlcolor=blue,breaklinks=true}

\makeatletter
\def\grd@save@target#1{%
  \def\grd@target{#1}}
\def\grd@save@start#1{%
  \def\grd@start{#1}}
\tikzset{
  grid with coordinates/.style={
    to path={%
      \pgfextra{%
        \edef\grd@@target{(\tikztotarget)}%
        \tikz@scan@one@point\grd@save@target\grd@@target\relax
        \edef\grd@@start{(\tikztostart)}%
        \tikz@scan@one@point\grd@save@start\grd@@start\relax
        \draw[minor help lines,magenta] (\tikztostart) grid (\tikztotarget);
        \draw[major help lines] (\tikztostart) grid (\tikztotarget);
        \grd@start
        \pgfmathsetmacro{\grd@xa}{\the\pgf@x/1cm}
        \pgfmathsetmacro{\grd@ya}{\the\pgf@y/1cm}
        \grd@target
        \pgfmathsetmacro{\grd@xb}{\the\pgf@x/1cm}
        \pgfmathsetmacro{\grd@yb}{\the\pgf@y/1cm}
        \pgfmathsetmacro{\grd@xc}{\grd@xa + \pgfkeysvalueof{/tikz/grid with coordinates/major step}}
        \pgfmathsetmacro{\grd@yc}{\grd@ya + \pgfkeysvalueof{/tikz/grid with coordinates/major step}}
        \foreach \x in {\grd@xa,\grd@xc,...,\grd@xb}
        \node[anchor=north] at (\x,\grd@ya) {\pgfmathprintnumber{\x}};
        \foreach \y in {\grd@ya,\grd@yc,...,\grd@yb}
        \node[anchor=east] at (\grd@xa,\y) {\pgfmathprintnumber{\y}};
      }
    }
  },
  minor help lines/.style={
    help lines,
    step=\pgfkeysvalueof{/tikz/grid with coordinates/minor step}
  },
  major help lines/.style={
    help lines,
    line width=\pgfkeysvalueof{/tikz/grid with coordinates/major line width},
    step=\pgfkeysvalueof{/tikz/grid with coordinates/major step}
  },
  grid with coordinates/.cd,
  minor step/.initial=.2,
  major step/.initial=1,
  major line width/.initial=2pt,
}
\makeatother

\usepackage{thmtools}
\usepackage{thm-restate}
\usepackage{etoolbox}
\makeatletter
\def\problem@s{}
\newcounter{problems@cnt}

\newcommand{\allproblems}{\problem@s}
\makeatother

\usepackage{amsmath,amsfonts,amssymb,array,graphicx,mathtools,multirow,bm,times,tcolorbox,relsize,booktabs,mathrsfs}
\usepackage[utf8]{inputenc}
\usepackage[T1]{fontenc}
\usepackage{qcircuit}
\usepackage{authblk} 
\usepackage[paperwidth=199.8mm,paperheight=297mm,centering,hmargin=20mm,vmargin=2cm]{geometry}
\usepackage[ruled, vlined, linesnumbered]{algorithm2e}
\usepackage{pifont} 
\usepackage{optidef} 
\usepackage[caption=false]{subfig} 

\definecolor{colortwo}{rgb}{0.4,0.77,0.17}
\definecolor{colorthree}{rgb}{0.01,0.51,0.93}

\pgfplotsset{compat=1.18}

\nc\bmu{{ \mathbf{u} }}
\nc\bmv{{ \mathbf{v} }}
\nc\bmp{{ \mathbf{p} }}
\nc\bmq{{ \mathbf{q} }}
\nc\bmone{{ \mathbf{1} }}
\allowdisplaybreaks
\begin{document}
\title{One-shot manipulation of coherence in dynamic \\quantum resource theory

}
\author[1]{Yu Luo\thanks{luoyu@snnu.edu.cn}}
\affil[1]{\small College of Computer Science, Shaanxi Normal University, Xi'an, 710062, China}
\date{\today}
\maketitle

\begin{abstract}
A fundamental problem in quantum information is to understand the operational significance of quantum resources. Quantum resource theories (QRTs) provide a powerful theoretical framework that aids in the analysis and comprehension of the operational meaning of these resources. Early resource theories primarily focused on analyzing static quantum resources. Recently, interest in the study of dynamic quantum resources has been growing. In this paper, we utilize superchannel theory to describe the dynamic resource theory of quantum coherence. In this dynamic resource theory, we treat classical channels as free channels and consider two classes of free superchannels that preserve channel incoherence [maximally incoherent superchannels (MISCs) and dephasing-covariant incoherent superchannels (DISCs)] as free resources. We regard the quantum Fourier transform as the golden unit of dynamic coherence resources. We first establish the one-shot theory of dynamic coherence cost and dynamic coherence distillation, which involves converting the quantum Fourier transform into an arbitrary quantum channel using MISCs and DISCs. Next, we introduce a class of free superchannels known as $\delta$-MISCs, which asymptotically generate negligible dynamic coherence. Finally, we provide upper and lower bounds for the one-shot catalytic dynamic coherence cost of quantum channels under the action of these $\delta$-MISC superchannels.
\end{abstract}
\section{Introduction} 
Quantum technologies offer significant advantages over classical computation in areas such as integer factorization~\cite{Shor1999,Ahnefeld2022}, quantum system simulation~\cite{Feynman2018simulating}, and quantum information processing~\cite{Fang2020,Fang2022,Regula2022,Takagi2020}. These advantages arise from the utilization of resources unique to quantum systems, such as entanglement~\cite{Horodecki2009,Vedral1997,PhysRevA.57.1619}, coherence~\cite{Streltsov2017a}, and magic states~\cite{Veitch2014}, among others. To investigate the potential operational advantages provided by these resources, quantum resource theoris (QRTs) were developed. In general, a QRT is characterized by a set of free states and a corresponding set of free operations that preserve the free states. States that do not belong to the set of free states are considered to possess resources~\cite{Brandao15}. For example, in the QRT of entanglement, the free states can be considered to be separable states, and the free operations are local operations and classical communication. Through these two main ingredients (free states and free operations), the QRT can quantitatively analyze the amount of resources present in quantum states and understand their operational significance~\cite{Takagi2022,Liu2019a,Streltsov2017,Chitambar2019}. It is worth noting that the advantages of quantum technologies depend not only on the resources possessed by a quantum state, but more on the amount of resources possessed by a quantum operation. For example, the key factor of Shor's algorithm is the application of the quantum Fourier transform gate. Therefore, studying QRT solely from the state level is incomplete, and a clear structure is also needed to study the resource theory of operations (dynamic resources). Furthermore, from the perspective of dynamic QRT (for comparison, we refer to state-based QRT as static QRT), a quantum state can also be viewed as a quantum channel with no input but a constant output, which serves as a representation of a quantum mechanical preparation apparatus. In this sense, dynamic QRT not only encompasses static QRT, but also possesses profound implications for exploring quantum advantages~\cite{Gour2019}. For the aforementioned reasons, researchers have discussed various topics from the perspective of dynamic resource theory and the resource theory of operations, including channel~\cite{PhysRevResearch.2.012035Liu-channel,Jencova2021general}, dynamic entanglement~\cite{Gour_PhysRevLett.125.180505,Gour2021-PhysRevA.103.062422}, dynamic coherence~\cite{Saxena20}, magic channels~\cite{Wang2019quantifying,Saxena2022}, informational non-equilibrium preservability~\cite{Stratton2024}, communication~\cite{Kristjansson2020resource,Hseih2021-PRXQuantum.2.020318,Takagi2020}, measurements~\cite{Skrzypczyk2019-PhysRevLett.122.140403,Tendick2023distancebased-measure}, measurement incompatibility~\cite{PhysRevLett.124.120401Buscemi-imcompatibility,RevModPhys.95.011003-imcompatible}, instrument incompatibility~\cite{Mitra2023-PhysRevA.107.032217,Ji2024-PRXQuantum.5.010340}, measurement sharpness~\cite{mitra2022quantifying,Buscemi2024completeoperational}, multi-time processes~\cite{Berk2021resourcetheoriesof}, quantum memories~\cite{Rosset2018,yuan2021universal}, causal connection~\cite{Milz2022resourcetheoryof}, and nonlocality~\cite{Buhrman2010-RevModPhys.82.665,Forster2009-PRL.102.120401,popescu1994quantum}. Similar to static QRT, dynamic QRT also consists of two main components: free channels (analogous to free states) and free superchannels (analogous to free operations). Naturally, the significance of these dynamic resources in operational tasks has attracted the attention of researchers~\cite{Saxena20,Gour2019,Regula2022,PhysRevLett.127.060402dynamicalQRT,Takagi2022,Theurer2019,Gour_PhysRevLett.125.180505}.

In this paper, we will investigate the concepts and related issues of QRT of operations, using coherence as a specific example. We will employ the quantum Fourier transform as a fundamental building block, which serves a role analogous to that of a maximally coherent state in static coherence resource theory. This resource theory, which uses operations rather than states as its fundamental elements, aligns more naturally with the name "dynamic QRT of coherence"~\cite{KIM-one-shot-2021,Gour2019}. Since classical channels cannot generate any quantum coherence, we extend the static QRT of coherence to the dynamic case by treating classical channels as free resources~\cite{Saxena20}. It is worth noting that in some studies focused on the coherence of operations, maximally incoherent operations (MIOs), which include classical channels, are often regarded as free channels~\cite{diaz2018-using,zhao2024probabilistic}. However, this raises certain issues: When using MIOs to detect the resources of a given state, the detection process itself requires resource consumption, leading to the question of whether MIOs can truly be considered "free"~\cite{Saxena20,Theurer2019,Egloff18}. Additionally, the identity channel, as a special case of a MIO, can also be regarded as a resource. This is because a physical system often undergoes decoherence processes, while the identity channel preserves the coherence of quantum memory, thus qualifying it as a resource channel~\cite{Saxena20,Simnacher_PhysRevA.99.062319,Rosset2018}. Moreover, a channel is classical if and only if its renormalized Choi matrix is an incoherent state (free state). This is connected to the dynamic entanglement theory, which considers separable operations to be free channels~\cite{Gour_PhysRevLett.125.180505}, as well as the magic channel theory, which regards completely stabilizer-preserving operations as free channels~\cite{Saxena2022}. Furthermore, it provides insights into how the so-called Choi-defined resource theories can be extended to dynamic resources~\cite{Zanoni2024choi}.

We also introduce two sets of free superchannels: maximally incoherent superchannels (MISCs) and dephasing-covariant incoherent superchannels (DISCs) to manipulate quantum channels. Since there is no physical restriction on such sets of free superchannels, they are useful in finding fundamental limitations to the ability of a quantum channel to generate coherent states. Based on the two types of free superchannels, MISCs and DISCs, we investigate the dynamic cost problem and the dynamic distillation problem of quantum channels in a one-shot setting. The one-shot dynamic cost quantifies the amount of resources consumed by the quantum Fourier transform when simulating the target channel under the action of free superchannels. The one-shot dynamic distillation quantifies the amount of resources consumed when using any channel to simulate the quantum Fourier transform under the action of free superchannels. Additionally, we introduce catalytic channels and study the limiting behavior of dynamic coherence cost under a class of superchannels called $\delta$-MISCs, where $\delta$-MISCs are a class of superchannels that can generate a small amount of dynamic coherence when acting on classical channels. 


This paper is organized as follows. In Sec.~\ref{sec:Preliminary}, we first introduce the necessary notation and definitions we need. In Sec.~\ref{sec:one-shot cost}, we provide the upper and lower bounds of the one-shot dynamic coherence cost under MISCs and DISCs. In Sec.~\ref{sec:one-shot distillation}, we provide the upper bound of one-shot dynamic coherence distillation under MISCs and DISCs. In Sec.~\ref{sec:one-shot catalytic cost}, we provide the upper and lower bounds of the one-shot catalytic dynamic coherence cost under $\delta$-MISCs. We summarize our results in Sec.~\ref{sec:conclusion}.
\section{Preliminary} \label{sec:Preliminary}
\paragraph{Notations.} 
Throughout this paper, we adopt most of the notations used in Ref.~\cite{Gour2019a}. We take the logarithm to base 2 and all Hilbert spaces $\mathscr{H}$ considered are finite dimensional. We will use \(A_0, A_1, B_0, B_1\), etc., to represent static systems and their corresponding Hilbert spaces. A replica of a system is represented using a tilde symbol. For instance, the system \(\tilde{A}_0\) denotes a replica of \(A_0\), and \(\tilde{A}_0 \tilde{B}_0\) denotes a replica of the composite system \(A_0 B_0\). This indicates that \(|\tilde{A}_0| = |A_0|\) and \(|\tilde{A}_0 \tilde{B}_0| = |A_0 B_0|\), where \(|A_0|\) represents the dimension of the system \(A_0\). The collection of all bounded operators on system $A_1$ will be denoted by $\textrm{B}(A_1)$, and all density matrices will be denoted by $\textrm{D}(A_1)$. Density matrices will be denoted using lowercase Greek letters such as \(\rho\), \(\sigma\), \(\tau\), and so on. The set of all linear maps from $\textrm B(A_0)$ to $\textrm B(A_1)$ will be denoted as $\textrm{L}(A_0\to A_1)$, among which all completely-positive and trace-preserving maps ($\CPTP$) are denoted as $\textrm{CPTP}(A_0\to A_1)$.  A CPTP map is also called a quantum channel. We will use calligraphic letters (\( \mathcal{E}, \mathcal{M}, \mathcal{N}, \mathcal{P}, \mathcal{Q} \), etc.) to denote quantum channels. The action of a quantum channel will usually be denoted by parentheses, as in \( \mathcal{N}(\rho) \). The normalized Choi matrix (or Choi state) of a quantum channel $\mathcal{N}_{\tilde{A}_0\to A_1}\in\textrm{CPTP}(\tilde{A}_0\to A_1)$ is given by $ J^{\mathcal{N}}=\text{id}_{A_0}\otimes\mathcal{N}_{\tilde{A}_0\to A_1}(\phi^+_{A_0\tilde{A}_0})$, where \( \text{id}_{A_0} \) is the identity map on system $A_0$ and $\phi^+_{A_0\tilde{A}_0}=\frac{1}{|A_0|}\sum_{i,j=0}^{|A_0|-1}\ketbra{ii}{jj}_{A_0\tilde{A}_0}$ is the maximally entangled state. 

In this context, we use \(A, B, C\), etc., to represent dynamic systems and their associated Hilbert spaces. We will assume that a system \( A \) has an associated input system \( A_0 \) and an output system \( A_1 \). Therefore, we will use shorthand notation, such as \( \textrm{L}(A) = \textrm{L}(A_0 \to A_1) \), \( \textrm{CPTP}(A) = \textrm{CPTP}(A_0 \to A_1) \), etc. A linear map from \(  
\textrm{L}(A) \) to \(\textrm{L}(B) \) is called a supermap, and the set of all such supermaps will be denoted by \(\textrm{L}(A \to B) \). We will use capital Greek letters like \( \Theta, \Omega \), etc., to denote supermaps. The action of a supermap will be represented with square brackets, as in \( \Theta[\mathcal{N}] \) and \( \Theta[\mathcal{E}] \). Superchannels are special supermaps that transform a quantum channel \( \mathcal{N}_{A_0 \to A_1} \) into another quantum channel \( \mathcal{M}_{B_0 \to B_1} \) through the expression \( \mathcal{M}_{B_0 \to B_1} = \mathcal{Q}_{A_1E \to B_1} \circ (\text{id}_E\otimes\mathcal{N}_{A_0 \to A_1} ) \circ \mathcal{P}_{B_0 \to A_0E} \), where \( \mathcal{P}_{B_0 \to A_0E} , \mathcal{Q}_{A_1E \to B_1} \) are the pre- and post-processing quantum channels, respectively~\cite{Chiribella2008,Gour2019a}. In the remainder of this paper, when the system on which a channel or superchannel acts is explicitly given, we will omit the subscript representing the system. For example, we will denote $\mathcal{N}_A\in\textrm{CPTP}(A_0\to A_1)$ or $\mathcal{N}_{ A_0\to A_1 }\in\textrm{CPTP}(A_0\to A_1)$ as $\mathcal{N}$, and $\Theta_{A\to B}$ as $\Theta$.
\paragraph{Static QRT of coherence.}
We begin by introducing the framework of the static QRT of coherence. We use the 2-tuple $\mathfrak{R}=(\mathfrak{F},\mathfrak{O})$, to denote a static QRT of coherence, the set of free states is denoted by $\mathfrak{F}$, and the set of free operations is denoted by $\mathfrak{O}$. In the static QRT of coherence, a free state $\sigma\in\mathfrak{F}$ (incoherent state) can be written as $\sigma=\sum_i\sigma_i\ketbra{i}{i}$ for a fixed basis $\{\ket{i}\}$. Variants of the free operations in the resource theory of coherence have been proposed. A CPTP map $\mathcal N$ is said to have MIOs, if $\mathcal N$ maps any incoherent state to an incoherent state~\cite{Aberg2006quantifyingsuperposition}. A CPTP map $\mathcal N$ is said to be a dephasing-covariant incoherent operation (DIO), if $\mathcal N\circ \mathcal{D}=\mathcal{D}\circ\mathcal{N}$, where $\mathcal{D}(\rho)=\sum_{i=1}^{d}\ket{i}\bra{i}\rho\ket{i}\bra{i}$ is a completely dephasing channel~\cite{Chitambar2016,Marvian2016}. 
A CPTP map $\mathcal N$ is said to be a detection incoherent operation (DI)~\cite{Theurer2019} (or nonactivating~\cite{Liu2017-PhysRevLett.118.060502}), if $\mathcal D\circ \mathcal{N}=\mathcal{D}\circ\mathcal{N}\circ\mathcal{D}$. Clearly, the set of DIOs is a subset of MIOs. Let $\textrm{D}(\mathscr{H})$ be the set of quantum states, the basic two requirements for a functional $C: \textrm{D}(\mathscr{H})\to \mathbb{R}^+$ being a coherence measure for $\mathfrak{R}$ are~\cite{Streltsov17_rmp}:\\
$\textbf{[A1]}$~(Non-negativity): $C(\rho)\geq0$ and the equality holds if and only if $\mathcal{\rho}\in\mathfrak{F}$;\\
$\textbf{[A2]}$~(Monotonicity): $C(\mathcal{\rho})\geq C(\mathcal{E}(\rho))$, where $\mathcal{E}\in\mathfrak{O}$ is a free operation.

\paragraph{Dynamic QRT of coherence.}\label{sec:dyn-qrt}
In this section, we will show that the framework of dynamic QRT of coherence. In a dynamic QRT, denoted as a 2-tuple $\hat{\mathfrak{R}}=(\hat{\mathfrak{F}},\hat{\mathfrak{O}})$, the set of free channels is denoted by $\hat{\mathfrak{F}}$, and the set of free superoperations is denoted by $\hat{\mathfrak{O}}$. Free channels $\hat{\mathfrak{F}}$ are those quantum channels that do not possess any resource, and free superoperations $\hat{\mathfrak{O}}$ are a subset of superchannels that transform free channels into free channels \cite{PhysRevResearch.2.012035Liu-channel}. Let $\textrm{CPTP}(\mathscr{H})$ be the set of quantum channels, the basic two requirements for a functional $\hat C: \textrm{CPTP}(\mathscr{H})\to \mathbb{R}^+$ being a channel coherence measure for $\hat{\mathfrak{R}}$ are~\cite{PhysRevResearch.2.012035Liu-channel}:\\
$\textbf{[B1]}$~(Non-negativity): $\hat C(\mathcal{N})\geq0$ and the equality holds if and only if $\mathcal{N}\in\hat{\mathfrak{F}}$;\\
$\textbf{[B2]}$~(Monotonicity): $\hat C(\mathcal{N})\geq \hat C(\Theta[\mathcal{N}])$, where $\Theta\in\hat{\mathfrak{O}}$ is a free superoperation.

In the dynamic resource theory of coherence, a free channel $\mathcal{N}\in \hat{\mathfrak{F}}$ is defined as a quantum channel that maps any incoherent states to another incoherent states. There are several types of channels can be considered as the free channel in the dynamic resource theory of coherence~\cite{Saxena20,diaz2018-using}. For instance: classical channel, MIOs and DIOs can be considered as the free channels. 

\begin{definition}[Classical channel~\cite{Saxena20}]
$\mathcal{N}\in\textrm{CPTP}(A_0\to A_1)$ is called a classical channel if it satisfies
\begin{equation}\label{eq:classicalchannel}
\mathcal{N}=\mathcal{D}_{A_1}\circ\mathcal{N}\circ\mathcal{D}_{A_0}
\end{equation}
where $\mathcal{D}_{A_k}(\rho)=\sum_{i=1}^{d_k}\ket{i}\bra{i}\rho\ket{i}\bra{i}$ is a completely dephasing channel for systems $A_k$ ($k=0,1$) in the incoherent bases $\{\ket{i}\}_{i=0}^{d_k-1}$.
\end{definition}
We will denote the set of classical channels that take system \(A_0\) to \(A_1\) by \(\textrm{C}(A_0 \to A_1)\). It is easy to see that the normalized Choi matrix $J^{\mathcal{N}}$ of a classical channel $\mathcal{N}$ is an incoherent state. Similar to the numerous free operations in the static QRT of coherence, there are various free superchannels in the dynamic QRT of coherence~\cite{Saxena20}. Here, we list two types of free superchannels studied in this paper:

\begin{definition}[Maximally incoherent superchannels (MISC)~\cite{Saxena20}]
    A superchannel $\Theta\in \hat{\mathfrak{O}}(A\to B)$ is said to be MISC if for any quantum channel $\mathcal{N}_A\in\textrm{C}(A_0\to A_1)$, $\Theta[\mathcal{N}_A]\in \textrm{C} (B_0\to B_1)$.
\end{definition}

\begin{definition}[Dephasing incoherent superchannels (DISC)~\cite{Saxena20}]
    A superchannel $\Theta\in\hat{\mathfrak{O}}(A\to B)$ is said to be DISC if 
    \begin{equation}
        \Delta_B\circ\Theta=\Theta\circ\Delta_B,
    \end{equation}
where $\Delta_B\in\hat{\mathfrak{O}}(B\to B)$ is the dephasing superchannel defined by $\Delta_B[\mathcal{N}_B]=\mathcal{D}_{B_1}\circ\mathcal{N}_B\circ\mathcal{D}_{B_0}$.
\end{definition}    

We denote the sets of MISC, DISC as $\textrm{MISC}(A\to B)$ and $\textrm{DISC}(A\to B)$, respectively. 


\paragraph{Measures of static and dynamic QRT of coherence.}\label{sec:meausres-dyn-qrt}
In this section, we will show that the resource measures of static and dynamic QRTs of coherence is used in this paper. 

For any quantum state $\rho$, the robustness of coherence for $\rho$ is defined as~\cite{Napoli2016-PhysRevLett.116.150502,Piani2016}
\begin{equation}
    C_R(\rho)=\min\{ s\geq0: \frac{1}{1+s}(\rho+s\sigma)\in\mathfrak{F}\},
\end{equation}
where $\sigma$ is a quantum state.
The log-robustness of coherence for $\rho$ is defined as
\begin{equation}
    LR(\rho)=\min_{\sigma\in\mathfrak{F}}D_{\max}(\rho||\sigma),
\end{equation}
where $D(\rho||\sigma)=\log\min\{\lambda\geq 0:\rho\leq\lambda\sigma\}$ is the max-relative entropy~\cite{Datta2009}. The condition $\rho \leq \lambda \sigma$ means that $\lambda \sigma - \rho$ is positive semi-definite.

For any quantum channel \(\mathcal{N}\in \text{CPTP}(A_0\to A_1)\), the robustness of coherence for $\mathcal{N}$ is defined as
\begin{align}\label{eq:ROC}
\hat{C}_R(\mathcal{N})=\min\{s\geq0: \mathcal{M}\in \text{CPTP}(A_0\to A_1),\\\nonumber
\frac{1}{1+s}(\mathcal{N}+s\mathcal{M})\in\textrm{C}(A_0\to A_1)
\}.    
\end{align}

The log-robustness of coherence for $\mathcal{N}$ is defined as
\begin{equation}
\hat{LR}(\mathcal{N})=\min_{ \mathcal{M}\in\textrm{C}(A_0\to A_1)}\hat{D}_{\max}(\mathcal{N}||\mathcal{M}),
\end{equation}
where $\hat{D}_{\max}(\mathcal{N}||\mathcal{M})=\log\min\{\lambda\geq 0:\mathcal{N}\leq\lambda\mathcal{M}\}$ is the max-relative entropy of the channel. The inequality $\mathcal{N} \leq \lambda \mathcal{M}$ is interpreted in terms of the completely positive ordering of superoperators, meaning that $\lambda \mathcal{M} - \mathcal{N}$ is a completely-positive map. Equivalently, the max-relative entropy can also be expressed as \(\hat{D}_{\max}(\mathcal{N}||\mathcal{M})=\log \min \{ \lambda \geq 0 : J^{\mathcal{N}} \leq \lambda J^{\mathcal{M}} \} \). The proof of this equivalent form can be found in Lemma 12 of Ref.~\cite{Wilde2020amortized}.

The smoothed version of the log-robustness of coherence is defined as follows~\cite{gour2024resources,Yu2024-PhysRevA.109.052413}:
\begin{equation}
    \hat{LR}_\epsilon(\mathcal N)=\min_{\mathcal{N}'\in B_{\epsilon}(\mathcal{N})  } \hat{LR} (\mathcal{N}'),
\end{equation}
where $\mathcal{N}' \in B_{\epsilon}(\mathcal{N})\iff\frac{1}{2}||\mathcal{N}-\mathcal{N}'||_{\diamond}\leq\epsilon$ and $||X_A||_{\diamond}=\max_{\rho_{AE}}\Tr|X_A\otimes\textit{id}_E(\rho_{AE})|$ is the diamond norm~\cite{PhysRevResearch.2.012035Liu-channel}. The inequality $||\Theta[\mathcal{N}_1] - \Theta[\mathcal{N}_2]||_{\diamond} \leq ||\mathcal{N}_1 -\mathcal{N}_2||_{\diamond}$, valid for any superchannel $\Theta$ and quantum channels $\mathcal{N}_1$ and $\mathcal{N}_2$~\cite{Gour2019a}, makes it straightforward to deduce that this smoothed quantity $\hat{LR}_\epsilon(\mathcal N)$ also qualifies as a dynamic resource monotone. 

\section{One-shot dynamic coherence cost under MISC and DISC}\label{sec:one-shot cost}
Resource cost stands out as an important subclass of resource manipulation tasks. Dynamic coherence cost is a protocol to transform the quantum Fourier transform (QFT) $\mathcal{F}_{d}$ into a target channel $\mathcal{N}$ using free superchannels, in which a quantum Fourier transforms is transformed into the desired channel. The optimal performance of this task is characterized by the one-shot dynamic coherence cost. Formally, we have following definition:
\begin{definition}\label{Def.one-shot-cost}
Given $\epsilon\geq0$, the one-shot dynamic coherence cost of a quantum channel $\mathcal{N}\in\textrm{CPTP}(A_0\to A_1)$ is defined as
\begin{equation}
c_{\epsilon,\hat{\mathfrak{O}}}^{(1)}(\mathcal{N})=\min \{\log d^2: \frac{1}{2}||\mathcal{N}-\Theta[\mathcal{F}_{d}]||_{\diamond}\leq\epsilon, \Theta\in\hat{\mathfrak{O}}(A'\to A) \},
\end{equation}
where $\hat{\mathfrak{O}}\in\{\textrm{MISC}, \textrm{DISC}\}$ and the QFT channel $\mathcal{F}_d(\cdot)=F_d(\cdot) F_d^{\dag}$ consists of the application of the QFT gate~\cite{Nielsen2010}:
\begin{equation}
    F_d=\frac{1}{\sqrt d}\sum_{j,k=0}^{d-1}e^{\frac{2\pi i}{d}jk} \ketbra{j}{k}.  
\end{equation}   
\end{definition}

The one-shot dynamic coherence cost can also be viewed as the minimum dimension of $\mathcal{F}_{d}$ needed to simulating a desired channel $\mathcal{N}$ by using some free pre-selected (post-selected) channels.

\begin{remark}
In dynamic QRTs, golden units are fundamental resource elements that serve as the most valuable and universally applicable building blocks within a given resource framework. They are often the maximally resourceful channels, meaning that any other resourceful channel can be generated from them using free superchannels. In Appendix~\ref{app:golden-units}, we show that any channel can be obtained from a QFT channel through a \textrm{MISC}, indicating that the QFT channel can be regarded as the golden unit in the dynamic QRT of coherence where the free superchannels are \textrm{MISCs}.
In the dynamic resource theory where free superchannels belong to \(\textrm{DISCs}\), directly proving that QFT channels can be converted into any other channels is challenging. However, it is reasonable to regard them as the golden units of this dynamic resource theory of coherence for the following reasons:  First, as shown in Appendix~\ref{app:maximal-replacement}, the dephasing log-robustness of the QFT channel (see Definition \ref{def:dephasing-logrob}) attains the maximal value of this measure, \(\log d^2\), which is a necessary condition for being a golden unit. 

Second, in addition to the QFT channel, maximal replacement channels \(\mathcal{R}_d(\cdot) = \Tr(\cdot) \psi^+_{d}\) could also be considered as potential golden units for this dynamic resource theory, where \(\psi^+_{d} = \frac{1}{d} \sum_{j,k=0}^{d-1} \ketbra{j}{k}\) is the maximally coherent state. However, as shown in Appendix~\ref{app:maximal-replacement}, both the log-robustness and the dephasing log-robustness of maximal replacement channels are \(\log d\), whereas the log-robustness and dephasing log-robustness of QFT channels are \(\log d^2\).  
This indicates that maximal replacement channels do not meet the \textit{necessary condition} for being golden units. Additionally, in Appendix~\ref{app:QRT-convert-Replacement}, we construct a free superchannel \( \Theta \in \textrm{DISC} \) such that \( \Theta[\mathcal{F}_d] = \mathcal{R}_d \).
In summary, treating QFT channels as the golden units for the two types of free superchannels (\(\textrm{MISCs}\) and \(\textrm{DISCs}\)) is a well-justified choice.
\end{remark}


The following result provides the upper and lower bounds of the one-shot dynamic coherence cost under MISCs:
\begin{shaded}
\begin{theorem}\label{the:cost-misc-upp-lower}
Let $d_0=\min\{d\in\mathbb{N}: \log d^2\geq   \hat{LR}_{\epsilon}(\mathcal{N})\}$. The one-shot dynamic coherence cost under MISC is bounded as follows
\begin{equation}\label{eq:misc-equality}
\hat{LR}_{\epsilon}(\mathcal N)\leq c_{\epsilon,\textrm{MISC}}^{(1)}(\mathcal N)< \hat{LR}_{\epsilon}(\mathcal N)+\log(\frac{d_0}{d_0-1})^2.
\end{equation}  
\end{theorem}
\end{shaded}
Before proving Theorem \ref{the:cost-misc-upp-lower}, we need to prove following lemma:
\begin{shaded}
\begin{lemma}\label{lemma:cost-f_d}
For any $\Theta\in\hat{\mathfrak{O}}(A'\to A)$, if $\frac{1}{2}||\mathcal N-\Theta[\mathcal{F}_d]||_{\diamond}\leq\epsilon$, then we have
    \begin{equation}
        \hat{LR}_{\epsilon}(\mathcal{N})\leq\hat{LR}(\mathcal{F}_d)=\log d^2.
    \end{equation}
\end{lemma}    
\end{shaded}
\begin{proof}
Let $\mathcal{E}^*\in\textrm{C}(A'_0\to A'_1)$ be the optimal quantum channel for $\hat{LR}(\mathcal{F}_d)$. Then, we have:
\begin{eqnarray}
   \hat{LR}_{\epsilon}(\mathcal{N})
   &\leq&\nonumber
   \hat{LR}(\Theta[\mathcal{F}_d])   
   \\\nonumber&=&  
   \min_{\mathcal{M}\in\textrm{C}(A_0\to A_1)}\hat{D}_{\max}(\Theta[\mathcal{F}_d]||\mathcal{M})
   \\\nonumber &\leq& 
   \hat{D}_{\max}(\Theta[\mathcal{F}_d]||\Theta[\mathcal{E}^*])
   \\\nonumber &\leq&
   \hat{D}_{\max}(\mathcal{F}_d||\mathcal{E}^*)
   \\\nonumber &=&
   \hat{LR}(\mathcal{F}_d)
   \\ &=&
   \log d^2,
\end{eqnarray}
where the fact that the second inequality holds follows from $\Theta[\mathcal{E}^*]\in\textrm{C}(A_0\to A_1)$.
\end{proof}

Now, we will prove the Theorem \ref{the:cost-misc-upp-lower}:

\begin{proof}\textbf{of Theorem \ref{the:cost-misc-upp-lower}}
For the lower bound, setting
\begin{equation}
    d_0=\min\{d\in\mathbb{N}: \log d^2\geq   \hat{LR}_{\epsilon}(\mathcal{N})\}. 
\end{equation}
Then, lemma \ref{lemma:cost-f_d} shows that 
\begin{equation}
 c_{\epsilon,\hat{\mathfrak{O}}}^{(1)}(\mathcal N) \geq \log d_0^2   \geq \hat{LR}_{\epsilon}(\mathcal N),
\end{equation}  
  For the upper bound, let $\mathcal{N}_\epsilon\in \textrm{CPTP}(A_0\to A_1)$ satisfies that $\frac{1}{2}||\mathcal{N}_\epsilon-\mathcal{N}||_{\diamond}\leq\epsilon$ and $\mathcal{P}$ be the classical channel that achieves $\hat{LR}_\epsilon(\mathcal{N})$, which implies $\mathcal{N}_\epsilon\leq 2^{\hat{LR}_\epsilon(\mathcal{N})}\mathcal{P}$. We can construct a supermap $\Theta_{A\to B}$ as follows
\begin{eqnarray}
    \Theta[\mathcal{N}] 
    &=&\nonumber
    \frac{d_0^2}{d_0^2-1}  (\Tr(J^{\mathcal{F}_{d_0}}J^{\mathcal{N}})-\frac{1}{d_0^2})\mathcal{N}_{\epsilon}    +     
    \frac{d_0^2}{d_0^2-1}  (1-\Tr(J^{\mathcal{F}_{d_0}}J^{\mathcal{N}}))\mathcal{P}    
    \\ &=&
    \frac{d_0^2}{d_0^2-1}  (1-\Tr(J^{\mathcal{F}_{d_0}}J^{\mathcal{N}}))  (\mathcal{P}-\frac{1}{d_0^2}\mathcal{N}_{\epsilon} )    +    
    \Tr(J^{\mathcal{F}_{d_0}}J^{\mathcal{N}})    \mathcal{N}_\epsilon,
\end{eqnarray}
where \( J^{\mathcal{F}_{d_0}} \) and \( J^{\mathcal{N}} \) denote the normalized Choi matrix of the QFT channel \( \mathcal{F}_{d_0} \) and the quantum channel \( \mathcal{N} \), respectively.
From the definition of $d_0$, we have $\mathcal{P}-\frac{1}{d_0^2}\mathcal{N}_{\epsilon}\geq0$, which implies $\Theta$ is a superchannel~\cite{Gour2019a}. Since for any classical channel $\mathcal{Q}\in \textrm{C}(A_0\to A_1)$, $\Tr(J^{\mathcal{F}_{d_0}}J^{\mathcal{Q}})=\frac{1}{d_0^2}$, so $\Theta[\mathcal{Q}]\in\textrm{C}(B_0\to B_1)$. Thus, $\Theta\in \textrm{MISC}$. Meanwhile, we have $\Theta[\mathcal{F}_{d_0}] = \mathcal{N}_\epsilon$. This means that $\log d_0^2$ is the achievable rate of $c_{\epsilon,\textrm{MISC}}^{(1)}(\mathcal N)$. By the definition of $d_0$, we have 
\begin{equation}
\log (d_0-1)^2<\hat{LR}_\epsilon(\mathcal{N})\implies
\log d_0^2+\log (d_0-1)^2<\hat{LR}_\epsilon(\mathcal{N})+\log d_0^2 \implies
\log d_0^2< \hat{LR}_\epsilon(\mathcal{N})+\log (\frac{d_0}{d_0-1})^2.
\end{equation}
\end{proof}

Although \(\log d^2_0\) is an achievable rate for \(c^{(1)}_{\epsilon,\textrm{MISC}}(\mathcal{N})\), the upper bound determined by Eq.~(\ref{eq:misc-equality}) is particularly useful in calculating the regularization of the smoothed coherence cost, as defined in Eq.~(\ref{def:asy-log-rob}).

To establish the bounds of the one-shot dynamic coherence cost under DISC, we need to define the following quantity:
\begin{definition}\label{def:dephasing-logrob}
For any quantum channel $\mathcal{N}\in \textrm{CPTP}(A_0\to A_1)$, its dephasing log-robustness of coherence is defined as
\begin{equation}
    \hat{LR}_{\Delta}(\mathcal{N})=\hat{D}_{\max} (\mathcal{N}||\Delta[\mathcal{N}]),
\end{equation}
where $\Delta$ is the dephasing superchannel. Its smoothed version is defined as
\begin{equation}
    \hat{LR}_{\epsilon,\Delta}(\mathcal{N})=\min_{\mathcal{N}'\in B_{\epsilon}(\mathcal{N})  } \hat{D}_{\max} (\mathcal{N}'||\Delta[\mathcal{N}']).
\end{equation}
\end{definition}

The following result provides the upper and lower bounds of the one-shot dynamic coherence cost under DISCs:
\begin{shaded}
\begin{theorem}\label{the:log-rob-lower-DISC}
Let $d_0=\min\{ d\in\mathbb{N}: \log d^2\geq \hat{LR}_{\epsilon,\Delta}(\mathcal{N})\}$, the one-shot dynamic coherence cost under DISC is bound as follows
\begin{equation}\label{eq:disc-equality}
\hat{LR}_{\epsilon,\Delta}(\mathcal N)\leq c_{\epsilon,\textrm{DISC}}^{(1)}(\mathcal N) < \hat{LR}_{\epsilon,\Delta}(\mathcal N)+\log (\frac{d_0}{d_0-1})^2.
\end{equation}  
\end{theorem}     
\end{shaded}
\begin{proof}
For the lower bound, suppose a free superchannel $\Theta\in \textrm{DISC}(A'\to A)$ and a QFT channel $\mathcal{F}_d$ exist such that $\frac{1}{2}||\mathcal N-\Theta[\mathcal{F}_d]||_{\diamond}\leq\epsilon$. Then, we have
\begin{eqnarray}
   \hat{LR}_{\epsilon,\Delta}(\mathcal{N})
   &\leq&\nonumber
   \hat{LR}_{\Delta}(\Theta[\mathcal{F}_d])   
   \\\nonumber&=&  
   \hat{D}_{\max}(\Theta[\mathcal{F}_d]||\Delta\circ\Theta[\mathcal{F}_d])
   \\\nonumber &=& 
  \hat{D}_{\max}(\Theta[\mathcal{F}_d]||\Theta\circ\Delta[\mathcal{F}_d])   
   \\\nonumber &\leq&
   \hat{D}_{\max}(\mathcal{F}_d||\Delta[\mathcal{F}_d])
   \\\nonumber &=&
   \log d^2
   \\ &=&
   c_{\epsilon,\textrm{DISC}}^{(1)}(\mathcal N),
\end{eqnarray}
where the fact that the second equality holds follows from $\Theta\in \textrm{DISC}$, the fact that the second inquality holds follows from the monotonicity of the max-relative channel divergence under superchannels, and the third equality holds because $ \hat{D}_{\max}(\mathcal{F}_d||\Delta[\mathcal{F}_d])=D_{\max}(J^{\mathcal{F}_d}||J^{\Delta[\mathcal{F}_d]})=D_{\max}(J^{\mathcal{F}_d}||\mathcal{D}(J^{\mathcal{F}_d}))=D_{\max}(J^{\mathcal{F}_d}||\frac{I}{d^2})=\log d^2$ where $\mathcal{D}(\cdot)=\sum_i\ketbra{i}{i}\cdot\ketbra{i}{i}$ is the completely dephasing channel.  

For the upper bound, let $\mathcal{N}_\epsilon\in \textrm{CPTP}(A_0\to A_1)$ satisfies $\hat{LR}_{\epsilon,\Delta}(\mathcal{N})=\hat{D}_{\max}( \mathcal{N}_\epsilon||\Delta[\mathcal{N}_\epsilon] )$, which implies $\mathcal{N}_{\epsilon}\leq 2^{\hat{LR}_{\Delta} ( \mathcal{N})}\Delta[\mathcal{N}_{\epsilon}]$. Now, we can construct following supermap $\Theta$:
\begin{eqnarray}\label{eq:supermap-disc}
    \Theta[\mathcal{N}] 
&=&\nonumber
\frac{d_0^2}{d_0^2-1} (\Tr(J^{\mathcal{F}_{d_0}} J^{\mathcal{N}})-\frac{1}{d_0^2}) \mathcal{N}_{\epsilon} +
\frac{d_0^2}{d_0^2-1}  (1-\Tr(J^{\mathcal{F}_{d_0}} J^{\mathcal{N}})) \Delta[\mathcal{N}_\epsilon]
\\&=&
\frac{d_0^2}{d_0^2-1} (1- \Tr(J^{\mathcal{F}_{d_0}} J^{\mathcal{N}}))
(\Delta[\mathcal{N}_\epsilon]-\frac{1}{d_0^2}\mathcal{N}_{\epsilon}) +
\Tr(J^{\mathcal{F}_{d_0}} J^{\mathcal{N}}) \mathcal{N}_{\epsilon}.
\end{eqnarray}
From the definition of $d_0$, we have $\Delta[\mathcal{N}_\epsilon]-\frac{1}{d_0^2}\mathcal{N}_{\epsilon}\geq0$, which implies $\Theta$ is a superchannel~\cite{Gour2019a}. Now, we will show $\Theta\in \textrm{DISC}$ by noting that
\begin{eqnarray}
    \Theta\circ\Delta[\mathcal{N}]
&=&\nonumber
\frac{d_0^2}{d_0^2-1} (1- \Tr(J^{\mathcal{F}_{d_0}} J^{\Delta[\mathcal{N}]}))
(\Delta[\mathcal{N}_\epsilon]-\frac{1}{d_0^2}\mathcal{N}_{\epsilon}) +
\frac{d_0^2}{d_0^2-1}
(1-\Tr(J^{\mathcal{F}_{d_0}} J^{\Delta[\mathcal{N}]})) \Delta[\mathcal{N}_{\epsilon}]
,\\&=& 
\Delta[\mathcal{N}_{\epsilon}],
\end{eqnarray}
and
\begin{eqnarray}\label{eq:delta-delta}
    \Delta\circ\Theta[\mathcal{N}] 
&=&\nonumber
\frac{d_0^2}{d_0^2-1} (1- \Tr(J^{\mathcal{F}_{d_0}} J^{\mathcal{N}}))
(\Delta[\mathcal{N}_\epsilon]-\frac{1}{d_0^2}\Delta[\mathcal{N}_{\epsilon}]) +
\Tr(J^{\mathcal{F}_{d_0}} J^{\mathcal{N}}) \Delta[\mathcal{N}_{\epsilon}]
\\&=& 
\Delta[\mathcal{N}_{\epsilon}],
\end{eqnarray}
where we used $\Delta\circ\Delta=\Delta$ in the first equality in Eq.(\ref{eq:delta-delta}).
Moreover, feeding $\mathcal{F}_{d_0}$ into Eq.(\ref{eq:supermap-disc}), we find $\Theta[\mathcal{F}_{d_0}]=\mathcal{N}_\epsilon$, which implies that $\log d_0^2$ is an achievable rate. From the definition of $d_0$, we have
\begin{equation}
\log (d_0-1)^2<\hat{LR}_{\epsilon,\Delta}(\mathcal{N})\implies
\log d_0^2+\log (d_0-1)^2<\hat{LR}_{\epsilon,\Delta}(\mathcal{N})+\log d_0^2,
\end{equation}
which implies
\begin{equation}
\log d_0^2< \hat{LR}_{\epsilon,\Delta}(\mathcal{N})+\log (\frac{d_0}{d_0-1})^2.
\end{equation}
\end{proof}

Although \(\log d^2_0\) is an achievable rate for \(c^{(1)}_{\epsilon,\textrm{DISC}}(\mathcal{N})\), the upper bound determined by Eq.(\ref{eq:disc-equality}) is particularly useful for calculating the regularization of the smoothed coherence cost, as defined in Eq.(\ref{def:asy-deph-log-rob}).

With the results above, we will now provide an exact description of the regularization of the smoothed coherence cost under free superchannels. First, we present the definition of the regularization of the smoothed coherence cost under free superchannels.

The regularization of the smoothed coherence cost under $\hat{\mathfrak{O}}\in\{\textrm{MISC},\textrm{DISC}\}$ can be defined as follows
\begin{equation}
    c^{\infty}_{\hat{\mathfrak{O}}}(\mathcal{N})= \lim_{\epsilon\to0^{+}}\lim_{n\to\infty}\frac{1}{n}c_{\epsilon,\hat{\mathfrak{O}}}^{(1)}(\mathcal{N}^{\otimes n}).
\end{equation}
Meanwhile, the asymptotic log-robustness is defined as
\begin{equation}\label{def:asy-log-rob}
    \hat{LR}^{\infty}(\mathcal{N})= \lim_{\epsilon\to0^{+}}\liminf_{n\to\infty}\frac{1}{n}\hat{LR}_\epsilon(\mathcal{N}^{\otimes n}),
\end{equation}
and the asymptotic dephasing log-robustness is defined as
\begin{equation}\label{def:asy-deph-log-rob}
    \hat{LR}^{\infty}_{\Delta}(\mathcal{N})= \lim_{\epsilon\to0^{+}}\liminf_{n\to\infty}\frac{1}{n}\hat{LR}_{\epsilon,\Delta}(\mathcal{N}^{\otimes n}).
\end{equation}
\begin{shaded}
\begin{corollary}\label{coro:regu-smooth-CS}
The regularization of the smoothed coherence cost under MISCs is equal to the asymptotic log-robustness, i.e.,
\begin{equation}\label{eq:regu-misc}
c^{\infty}_{\textrm{MISC}}(\mathcal{N})=\hat{LR}^{\infty}(\mathcal{N}),
\end{equation}
and the regularization of the smoothed coherence cost under DISCs is equal to the asymptotic dephasing log-robustness, i.e.,
\begin{equation}\label{eq:regu-disc}
c^{\infty}_{\textrm{DISC}}(\mathcal{N})=\hat{LR}^{\infty}_{\Delta}(\mathcal{N}).
\end{equation}
\end{corollary}    
\end{shaded}
\begin{proof}
We need to prove only Eq.(\ref{eq:regu-misc}), as the proof of Eq.(\ref{eq:regu-disc}) follows a similar argument. From Eq.(\ref{eq:misc-equality}), it follows that
\begin{equation}
\hat{LR}_{\epsilon}({\mathcal N}^{\otimes n}) \leq c_{\epsilon, \mathrm{MISC}}^{(1)}(\mathcal N^{\otimes n}) < \hat{LR}_{\epsilon}(\mathcal N^{\otimes n}) + \log\left(\frac{d_n}{d_n-1}\right)^2,
\end{equation}    
where $d_n = \min\{d \in \mathbb{N} : \log d^2 \geq \hat{LR}_{\epsilon}(\mathcal{N}^{\otimes n})\}$. Since $d_n \geq 2$ for any $n \geq 1$, we have $\log\left(\frac{d_n}{d_n-1}\right)^2 \leq 2$. Therefore, we obtain:
\begin{equation}
\frac{1}{n}\hat{LR}_{\epsilon}({\mathcal N}^{\otimes n}) \leq \frac{1}{n}c_{\epsilon, \mathrm{MISC}}^{(1)}(\mathcal N^{\otimes n}) < \frac{1}{n}\hat{LR}_{\epsilon}(\mathcal N^{\otimes n}) + \frac{2}{n}.
\end{equation}   

Taking the limits as $n \to \infty$ and $\epsilon \to 0^+$ completes the proof of Eq.(\ref{eq:regu-misc}). 
\end{proof}
\section{One-shot dynamic coherence distillation under \textrm{MISCs} and \textrm{DISCs}}\label{sec:one-shot distillation}
Resource distillation is the other important subclasses of resource manipulation tasks. In this section, we investigate the maximal amount of dynamic coherence required to simulate a quantum Fourier transform under free superchannels. The formal definition is thus given as follows:

\begin{definition}
Given $\epsilon\geq0$, the one-shot dynamic coherence distillation of a quantum channel $\mathcal{N}\in\textrm{CPTP}(A_0\to A_1)$ is defined as
\begin{equation}\label{def:dynamic-coh-dis}
d_{\epsilon,\hat{\mathfrak O}}^{(1)}(\mathcal{N})=\max\{\log d^2:~ \frac{1}{2}||\Theta[\mathcal{N}]-\mathcal{F}_d||_{\diamond}\leq\epsilon, \Theta\in \hat{\mathfrak O}(A\to B)   \}.
\end{equation}
where $\hat{\mathfrak O}  \in\{ \textrm{MISC}, \textrm{DISC}   \}$.   
\end{definition}
Next, we provide another bound of one-shot dynamic coherence distillation based on the smoothed hypothesis testing relative entropy. For any density matrices $\rho$ and $\sigma$, the smoothed hypothesis testing relative entropy (sometimes it is also called as the min-relative entropy) is defined as \cite{wang-PhysRevLett.108.200501,Nuradha2024-SMOOTHEDfidelity} 
\begin{equation}
    H_\epsilon(\rho||\sigma)=-\log \min_{0\leq P\leq I}\{ \Tr[P\sigma]: \Tr[P\rho] \geq 1-\epsilon         \}.
\end{equation}
Its smoothed channel divergence can be defined as \cite{cooney2016strong,Takagi2022,diaz2018-using,Yuan-PhysRevA.99.032317}
\begin{equation}
    \hat{H}_\epsilon(\mathcal{N}||\mathcal{M})=\max_{\psi\in
    \textrm{D}(RA)} H_{\epsilon}(id_R\otimes\mathcal{N}(\psi)||id_R\otimes\mathcal{M}(\psi)),
\end{equation}
where the optimization is restricted to pure input states $\psi$ without loss of generality.
\begin{definition}
For any two quantum channels $\mathcal{N}\in \textrm{CPTP}(A_0\to A_1)$, the hypothesis-testing relative entropy of coherence of $\mathcal{N}$ is defined as

\begin{equation}
    C_{H}^{\epsilon}(\mathcal{N})=\min_{\mathcal{M}\in\textrm{C}(A_0\to A_1) }\hat{H}_\epsilon(\mathcal{N}||\mathcal{M}).
\end{equation}   
\end{definition}

First, we have following result: 
\begin{shaded}
\begin{proposition} \label{prop:1st-CH}
For any quantum channel $\mathcal{N}\in \textrm{CPTP}(A_0\to A_1) $, the hypothesis-testing relative entropy of coherence is a dynamic coherence monotone under the free superchannels $\Theta\in\hat{\mathfrak{O}}(A\to B)$. That is 
\begin{equation}
C_{H}^{\epsilon}(\Theta[\mathcal{N}]) \leq C_{H}^{\epsilon}(\mathcal{N}).   
\end{equation}
\end{proposition}    
\end{shaded}
\begin{proof}
   First, we need to prove that the hypothesis-testing relative entropy of channel divergence is monotonic under the freesuperchannel $\Theta[\mathcal{N}]= \mathcal{P}_{A_1E\to B_1}\circ id_E\otimes\mathcal{N} \circ \mathcal{Q}_{B_0\to A_0E}$, i.e.,
\begin{equation}\label{eq:Hypothesistesting-DPI}
\hat{H}_\epsilon(\Theta[\mathcal{N}]||\Theta[\mathcal{M}])\leq\hat{H}_\epsilon(\mathcal{N}||\mathcal{M})    
\end{equation}
holds for any quantum channel $\mathcal{N},\mathcal{M}\in\textrm{CPTP}(A_0\to A_1)$. This is because the following inequalities chain holds
\begin{eqnarray}
\hat{H}_\epsilon(\Theta[\mathcal{N}]||\Theta[\mathcal{M}])
\nonumber &=& H_{\epsilon}(id_R\otimes\Theta[\mathcal{N}](\psi^*_{RB_0})||id_R\otimes\Theta[\mathcal{M}](\psi^*_{RB_0}))
\\\nonumber &=& H_{\epsilon}(id_{R}\otimes\mathcal{P}_{A_1E\to B_1}\circ id_E\otimes\mathcal{N} \circ \mathcal{Q}_{B_0\to A_0E}(\psi^*_{RB_0})  
      ||     id_R\otimes\mathcal{P}_{A_1E\to B_1}\circ id_E\otimes\mathcal{M} 
\\\nonumber     &\circ& \mathcal{Q}_{B_0\to A_0E}(\psi^*_{RB_0}))
\\\nonumber &\leq& H_{\epsilon}(id_{RE}\otimes\mathcal{N} \circ \mathcal{Q}_{B_0\to A_0E}(\psi^*_{RB_0})  
      ||     id_{RE}\otimes\mathcal{M} \circ \mathcal{Q}_{B_0\to A_0E}(\psi^*_{RB_0}))
\\\nonumber &\leq& H_{\epsilon}(id_{RE}\otimes\mathcal{N} (\psi^*_{REA})  
      ||     id_{RE}\otimes\mathcal{M}(\psi^*_{REA}))
\\\nonumber &\leq& \max_{\psi\in\textrm{D}(REA)}H_{\epsilon}(id_{RE}\otimes\mathcal{N} (\psi)  
      ||     id_{RE}\otimes\mathcal{M}(\psi))      
\\ &=&\hat{H}_\epsilon(\mathcal{N}||\mathcal{M}),
\end{eqnarray}
where $\psi^*_{RB_0}$ is the optimal pure state for $\hat{H}_\epsilon(\Theta[\mathcal{N}]||\Theta[\mathcal{M}])$, and the first and second inequalities holds because the data-processing inequality holds for the hypothesis-testing relative entropy~\cite{wang-PhysRevLett.108.200501}. Then, we have
\begin{eqnarray}
C_{H}^{\epsilon}(\mathcal{N}) 
\nonumber &=& \hat{H}_{\epsilon}(\mathcal{N}||\mathcal{M}^*)
\\\nonumber &\geq& \hat{H}_{\epsilon}(\Theta[\mathcal{N}]||\Theta[\mathcal{M}^*])
\\\nonumber &\geq& \min_{\tilde{\mathcal{M}}\in\textrm{C}(B_0\to B_1)}\hat{H}_{\epsilon}(\Theta[\mathcal{N}]||\tilde{\mathcal{M}}])
\\ &=& C_{H}^{\epsilon}(\Theta[\mathcal{N}]),
\end{eqnarray}
where $\mathcal{M}^*\in \textrm{C}(B_0\to B_1)$ is the optimal classical channel for $C_{H}^{\epsilon}(\mathcal{N})$, the first inequality holds because Eq.(\ref{eq:Hypothesistesting-DPI}) holds, and the second inequality holds because $\Theta[\mathcal{M}^*]\in\textrm{C}(B_0\to B_1)$.
\end{proof}

The following result provides the upper bound of the one-shot dynamic coherence distillation under MISC:
\begin{shaded}
\begin{theorem}\label{The:one-shot-distill-upperbound}
For any quantum channel $\mathcal{N}\in \textrm{CPTP}(A_0\to A_1) $, the one-shot dynamic coherence distillation under \textrm{MISCs} is bound as follows
\begin{equation}
d_{\epsilon,\textrm{MISC}}^{(1)}(\mathcal{N}) \leq C_{H}^{2\epsilon}(\mathcal{N}).
\end{equation}
\end{theorem}    
\end{shaded}
Before proving this theorem, we need to prove the following lemma:
\begin{lemma}\label{lemma:F_in_S}
For any superchannel $\Theta\in\hat{\mathfrak{O}}(A\to B)$, quantum channel $\mathcal{N}\in\textrm{CPTP}(A_0\to A_1)$ and pure state $\psi_{RB_0}\in \textrm{D}(RB_0)$, we have $0\leq id_R\otimes\mathcal{F}_{d,B}(\psi_{RB_0})\leq I$ and $\Tr [id_R\otimes\mathcal{F}_{d,B}(\psi_{RB_0})id_R\otimes\Theta[\mathcal{N}](\psi_{RB_0})]  \geq 1- 2\epsilon$.
\end{lemma}
\begin{proof}
First, it is easy to check that $0\leq id_R\otimes\mathcal{F}_{d,B}(\psi_{RB_0})\leq I$. Second,
\begin{eqnarray}
\Tr [id_R\otimes\mathcal{F}_{d,B}(\psi_{RB_0})id_R\otimes\Theta[\mathcal{N}](\psi_{RB_0})]
\nonumber 
&=& F\left(id_R\otimes\Theta[\mathcal{N}](\psi_{RB_0}),id_R\otimes\mathcal{F}_{d,B}(\psi_{RB_0})\right)
\\\nonumber &\geq& \left( 1- \frac{1}{2} ||id_R\otimes\Theta[\mathcal{N}_A](\psi_{RB_0})- 
id_R\otimes\mathcal{F}_{d,B}(\psi_{RB_0})||_1  \right)^2
\\\nonumber &\geq&   \left( 1-\max_{\psi_{RB_0}}  \frac{1}{2} ||id_R\otimes\Theta[\mathcal{N}_A](\psi_{RB_0})-id_R\otimes\mathcal{F}_{d,B}(\psi_{RB_0})||_1  \right)^2
\\\nonumber &=& \left( 1- \frac{1}{2} ||  \Theta[\mathcal{N}_A]-  \mathcal{F}_{d,B}||_{\diamond} \right)^2
\\\nonumber &\geq&    (1-\epsilon)^2
\\ &\geq&     1-2\epsilon,
\end{eqnarray}
where $F(\rho,\sigma)=||\sqrt{\rho}\sqrt{\sigma}||_1^2$ is the fidelity~\cite{Nielsen2010}, and the second inequality holds because of the Fuchs-van Graaf inequalities: 
$1-\sqrt{F(\rho,\sigma)}   \leq  \frac{1}{2}||\rho-\sigma||_1   \leq  \sqrt{1-F(\rho, \sigma)   }$~\cite{watrous2018}.  
\end{proof}

Next, we provide the proof of Theorem \ref{The:one-shot-distill-upperbound}:

\begin{proof}\textbf{of Theorem \ref{The:one-shot-distill-upperbound}}
Let $\Theta\in\textrm{MISC}(A\to B)$ be an optimal free superchannel such that $d_{\epsilon, \textrm{MISC}}^{(1)}(\mathcal{N})=\log d^2$, the following inequality chain holds
\begin{eqnarray}
C_{H}^{2\epsilon}(\mathcal{N})
\nonumber &\geq& 
C_{H}^{2\epsilon}(\Theta[\mathcal{N}])
\\\nonumber &=& 
\min_{\mathcal{M}_B\in\textrm{C}(B_0\to B_1)} \max_{\psi_{RB_0}} \max_{ 0\leq P_{RB_0}\leq I, \Tr[ P_{RB_0} id_R\otimes\Theta[\mathcal{N}](\psi_{RB_0})]\geq 1-2\epsilon   } 
-\log \Tr\left[P_{RB_0}id_R\otimes\mathcal{M}_B(\psi_{RB_0})
\right]
\\\nonumber&\geq& 
\min_{\mathcal{M}_B\in\textrm{C}(B_0\to B_1) } \max_{\psi_{RB_0}}  
-\log \Tr\left[id_R\otimes\mathcal{F}_{d,B}(\psi_{RB_0})id_R\otimes\mathcal{M}_B(\psi_{RB_0})
\right]
\\\nonumber &\geq&
\min_{\mathcal{M}_B\in\textrm{C}(B_0\to B_1) }  
-\log \Tr\left[id_R\otimes\mathcal{F}_{d,B}(\phi^+_{RB_0})id_R\otimes\mathcal{M}_B(\phi^+_{RB_0})
\right]
\\\nonumber &=&
\min_{\mathcal{M}_B\in\textrm{C}(B_0\to B_1) }   
-\log \Tr(J^{\mathcal{F}_{d}}J^{\mathcal{M}_B})
\\\nonumber &=&
\log d^2
\\ &=& 
d_{\epsilon,\textrm{MISC}}^{(1)}(\mathcal{N}),
\end{eqnarray}   
where the fact that the first line holds follows from the monotonicity of $C_{H}^{2\epsilon}(\mathcal{N})$ under \textrm{MISCs}, the third line holds because of Lemma~\ref{lemma:F_in_S}, and the fourth line holds because we choose the pure state \( \psi_{RB_0} \) to be the maximally entangled state \( \phi^+_{RB_0} = \frac{1}{d} \sum_{i,j=0}^{d-1} \ketbra{ii}{jj}_{RB_0} \). The fifth line holds because \( J^{\mathcal{F}_{d}} \) and \( J^{\mathcal{M}_B} \) are the normalized Choi matrices of \( \mathcal{F}_{d,B} \) and the classical channel \( \mathcal{M}_B \), respectively. Thus, the normalized Choi matrix \( J^{\mathcal{M}_B} \) can be viewed as an incoherent state and it is easy to check that for any incoherent state $\sigma$, $Tr[J^{\mathcal{F}_{d}}\sigma]=\frac{1}{d^2}$.
\end{proof}
\begin{definition}
For any quantum channel $\mathcal{N}\in \textrm{CPTP}(A_0\to A_1)$, its dephasing hypothesis-testing relative entropy of coherence is defined as
\begin{equation}
    C_{H,\Delta_A}^{\epsilon}(\mathcal{N})=\hat{H}_{\epsilon}(\mathcal{N}||\Delta_{A}[\mathcal{N}]),
\end{equation}
where $\Delta_{A}\in\hat{\mathfrak{O}}(A\to A)$ is the dephasing superchannel.
\end{definition}

\begin{shaded}
\begin{proposition} \label{prop:2ed-CH}
For any quantum channel $\mathcal{N}\in \textrm{CPTP}(A_0\to A_1) $, the dephasing hypothesis-testing relative entropy of coherence is a dynamic coherence monotone under the \textrm{DISCs}. In other words, for any $\Theta\in\textrm{DISC}(A\to B)$, we have
\begin{equation}
C_{H,\Delta_A}^{\epsilon}(\Theta[\mathcal{N}]) \leq C_{H,\Delta_A}^{\epsilon}(\mathcal{N}).   
\end{equation}
\end{proposition}    
\end{shaded}
\begin{proof}
  For any $\Theta\in \textrm{DISC}(A\to B)$, the following inequality holds:
  \begin{eqnarray}
    C_{H,\Delta_A}^{\epsilon}(\Theta[\mathcal{N}])
    &=&\nonumber
    \hat{H}_{\epsilon}(\Theta[\mathcal{N}]||\Delta_{B}\circ\Theta[\mathcal{N}])      
    \\\nonumber &=&  
    \hat{H}_{\epsilon}(\Theta[\mathcal{N}]||\Theta\circ\Delta_{A}[\mathcal{N}])  
    \\\nonumber &\leq&
       \hat{H}_{\epsilon}(\mathcal{N}||\Delta_{A}[\mathcal{N}])   
    \\\nonumber &=&
        C_{H,\Delta_A}^{\epsilon}(\mathcal{N}),    
  \end{eqnarray}
where the fact that the second equality holds follows from the definition of DISCs, and the inequality holds because the hypothesis-testing relative entropy of channel divergence is
monotonic under the freesuperchannel as shown in Proposition \ref{prop:1st-CH}.   
\end{proof}

The following result shows that the dephasing hypothesis-testing relative entropy of coherence is an upper bound of the one-shot dynamic coherence distillation under DISC:

\begin{shaded}
\begin{theorem}
For any quantum channel $\mathcal{N}\in \textrm{CPTP}(A_0\to A_1) $, the one-shot dynamic coherence distillation under DISC is bound as follows
\begin{equation}
d_{\epsilon,\textrm{DISC}}^{(1)}(\mathcal{N}) \leq  C_{H,\Delta}^{2\epsilon}(\mathcal{N}).
\end{equation}
\end{theorem}    
\end{shaded}

\begin{proof}
Let $\Theta\in\textrm{DISC}(A\to B)$ be an optimal free superchannel such that $d_{\epsilon}^{(1)}(\mathcal{N})=\log d^2$, the following chain of inequalities holds
\begin{eqnarray}
C_{H,\Delta}^{2\epsilon}(\mathcal{N})
\nonumber &\geq& 
C_{H,\Delta}^{2\epsilon}(\Theta[\mathcal{N}])
\\\nonumber &=& 
\hat{H}_{2\epsilon}(\Theta[\mathcal{N}]||\Delta_B\circ\Theta[\mathcal{N}])
\\\nonumber &=& 
\max_{\psi_{RB_0}} H_{2\epsilon}(id_R\otimes\Theta[\mathcal{N}](\psi_{RB_0})||id_R\otimes
 \Delta_B\circ\Theta[\mathcal{N}]   
(\psi_{RB_0}))
\\\nonumber&=& 
\max_{\psi_{RB_0}}  \max_{ 0\leq P_{RB_0}\leq I, \Tr[ P_{RB_0} id_R\otimes\Theta[\mathcal{N}](\psi_{RB_0})]\geq 1-2\epsilon   } 
-\log \Tr\left[P_{RB_0}id_R\otimes\Delta_B\circ\Theta[\mathcal{N}]   
(\psi_{RB_0})
\right]
\\\nonumber&\geq& 
\max_{\psi_{RB_0}}  -\log \Tr\left[id_R\otimes\mathcal{F}_{d,B}(\psi_{RB_0}) 
id_R\otimes\Delta_B\circ\Theta[\mathcal{N}] (\psi_{RB_0})
\right]
\\\nonumber &\geq&
-\log \Tr\left[id_R\otimes\mathcal{F}_{d,B}(\phi^+_{RB_0}) 
id_R\otimes\Delta_B\circ\Theta[\mathcal{N}] (\phi^+_{RB_0})
\right]
\\\nonumber &=&
\min_{\mathcal{M}_B\in\textrm{C}(B_0\to B_1) }   
-\log \Tr(J^{\mathcal{F}_{d}}J^{\Delta_B\circ\Theta[\mathcal{N}]})
\\\nonumber &=&
\log d^2
\\ &=& 
d_{\epsilon,\textrm{DISC}}^{(1)}(\mathcal{N}),
\end{eqnarray}   
where the first line follows from the monotonicity of $C_{H,\Delta}^{2\epsilon}(\mathcal{N})$ under DISCs, the fourth line follows from the definition of the hypothesis-testing relative entropy, the fifth line holds because we define $P_{RB_0} = \text{id}_R \otimes \mathcal{F}_{d,B}(\psi_{RB_0})$, and the sixth line holds because we choose the pure state \( \psi_{RB_0} \) to be the maximally entangled state \( \phi^+_{RB_0} = \frac{1}{d} \sum_{i,j=0}^{d-1} \ketbra{ii}{jj}_{RB_0} \). The eighth line holds because \( \Delta_B \circ \Theta[\mathcal{N}] \) is a classical channel, which follows from the fact that $\Delta \circ \Theta[\mathcal{N}] = \mathcal{D} \circ \Theta[\mathcal{N}] \circ \mathcal{D} = \mathcal{D} \circ \mathcal{D} \circ \Theta[\mathcal{N}] \circ \mathcal{D} \circ \mathcal{D} = \mathcal{D} \circ (\Delta \circ \Theta[\mathcal{N}]) \circ \mathcal{D}$. Thus, the normalized Choi matrix \( J^{\Delta_B \circ \Theta[\mathcal{N}]} \) is an incoherent state, and it is easy to check that for any incoherent state $\sigma$, \( \Tr[J^{\mathcal{F}_{d}} \sigma] = \frac{1}{d^2} \).
\end{proof}
\section{One-shot catalytic dynamic coherence cost under \texorpdfstring{$\delta$}{delta}-MISC}~\label{sec:one-shot catalytic cost}
To investigate the limitations of the dynamic coherence cost, we appropriately relax the conditions of the dynamic coherence cost protocol in Definition \ref{Def.one-shot-cost}. First, we allow the use of catalysts. Second, we consider a class of free superchannels that can generate a small amount of dynamic coherence when acting on classical channels, which we refer to as $\delta$-MISC. Based on these two relaxations, we will explore the one-shot dynamic coherence cost with the assistance of catalysis. We are now ready to present the formal definition of $\delta$-MISC:
\begin{definition}[$\delta$-MISC]
A superchannel $\Theta\in\hat{\mathfrak{O}}(A\to B)$ is said to be $\delta$-MISC if for any quantum channel $\mathcal{N}\in\textrm{C}(A_0\to A_1)$, $\hat{C}_R(\Theta[\mathcal{N}])\leq\delta$,
where \( \hat{C}_R \)  represents the robustness of coherence for a quantum channel.
\end{definition}
We denote the set of $\delta$-MISCs as $\delta$-$\textrm{MISC}(A\to B)$.

\begin{remark}
The intuition behind using $\hat{C}_R$ in the definition of $\delta$-MISCs stems from the static QRT~\cite{Brandao15,Chitambar2019,Hayashi2024generalized,Lami2024solution,Brandao2011one-shot,brandao2010reversible,xi2019coherence}. The counterpart of $\delta$-MISC in the static QRT is referred to as  $\epsilon$-resource non-generating maps, whose definition is directly based on the generalized robustness (see Sec III.C.3 in Ref.\cite{Chitambar2019}). These operations have been shown to be highly effective for studying asymptotic resource convertibility. Furthermore, in the study of the dynamic QRT of entanglement, the authors of Ref.\cite{KIM-one-shot-2021} similarly defined the free superchannels analogous to $\delta$-MISCs based on the robustness measure of the channel. For the above reasons, we use $\hat{C}_R$ as the fundamental quantity to define $\delta$-MISCs.
\end{remark}

The following results give the upper bounds of log-robustness $\hat{LR}(\mathcal{N})$ and its smoothed version $\hat{LR}_{\epsilon}(\mathcal{N})$ under a $\delta$-MISC:

\begin{shaded}
\begin{lemma}\label{lemma:LRbound}
 For any $\Theta\in\delta$-$\textrm{MISC}(A\to B)$ and quantum channel $\mathcal{N}\in \textrm{CPTP}(A_0\to A_1)$, the following holds
 \begin{equation}
\hat{LR}(\Theta[\mathcal{N}])\leq \hat{LR}(\mathcal{N})+\log(1+\delta).   
 \end{equation}
\end{lemma}    
\end{shaded}
\begin{proof}
Let $\mathcal{M}\in \textrm{CPTP}(A_0\to A_1)$ be the channel such that 
\begin{equation}
    \mathcal{G}:=\frac{\mathcal{N} + \hat{C}_R(\mathcal{N}) \mathcal{M}}  {1+\hat{C}_R(\mathcal{N})} \in\textrm{C}(A_0\to A_1),
\end{equation}
where $\hat{C}_R(\mathcal{N})$ is the robustness of coherence for $\mathcal{N}$. Using the linear property, we have
\begin{equation}\label{eq:pro-1}
\Theta[\mathcal{N}] + \hat{C}_R(\mathcal{N})   \Theta[\mathcal{M}]  =   
(1+\hat{C}_R(\mathcal{N}))   \Theta[\mathcal{G}].
\end{equation}
For the quantum channel $\Theta[\mathcal{G}]$, we can always find another quantum channel $\mathcal{P}\in \textrm{CPTP}(B)$ such that 
\begin{equation}
\mathcal{Q}:=\frac{\Theta[\mathcal{G}]+\hat{C}_R(\Theta[\mathcal{G}] ) \mathcal{P}}{1+\hat{C}_R(\Theta[\mathcal{G}] )}    \in \textrm{C}(B_0\to B_1).
\end{equation}
Then, we have
\begin{eqnarray}\label{eq:pro-2}
\Theta[\mathcal{G}]+\hat{C}_R(\Theta[\mathcal{G}] ) \mathcal{P}= (1+\hat{C}_R(\Theta[\mathcal{G}] )) \mathcal{ Q }
\end{eqnarray}
Combining Eq.(\ref{eq:pro-1}) with Eq.(\ref{eq:pro-2}), we obtain that
\begin{equation}
    \Theta[\mathcal{N}] + \hat{C}_R(\mathcal{N})   \Theta[\mathcal{M}] + (1+\hat{C}_R(\mathcal{N}))  \hat{C}_R(\Theta[\mathcal{G}] ) \mathcal{P}  =    
(1+\hat{C}_R(\mathcal{N}))  (1+\hat{C}_R(\Theta[\mathcal{G}] )) \mathcal{Q}.   
\end{equation}
Which implies that 
\begin{eqnarray}
    \Theta[\mathcal{N}] \nonumber &\leq&   
(1+\hat{C}_R(\mathcal{N}))  (1+\hat{C}_R(\Theta[\mathcal{G}] )) \mathcal{Q}
\\ &\leq& (1+\hat{C}_R(\mathcal{N}))  (1+\delta ) \mathcal{Q}.
\end{eqnarray}
Thus, we have
\begin{eqnarray}
    1+\hat{C}_R(\Theta[\mathcal{N}]) \leq (1+\hat{C}_R(\mathcal{N})) (1+\delta ). 
\end{eqnarray}
\end{proof}

\begin{shaded}
\begin{lemma}
For any $\Theta\in\delta\textrm{-MISC}(A\to B)$ and $\mathcal{N}\in\textrm{CPTP}(A_0\to A_1)$, the following holds
    \begin{equation}
        \hat{LR}_{\epsilon}( \Theta[ \mathcal{N}   ])\leq \hat{LR}_{\epsilon}(\mathcal{N})  +  \log(1+\delta).  
    \end{equation}
\end{lemma}    
\end{shaded}
\begin{proof}
Let $\mathcal{N}^*\in\textrm{CPTP}(A_0\to A_1)$ be the optimal argument of $\hat{LR}_{\epsilon}(\mathcal{N})$, such that $\hat{LR}_{\epsilon}(\mathcal{N})=\hat{LR}(\mathcal{N}^*)$. Then, the following inequalities hold:
\begin{eqnarray}
    \hat{LR}_{\epsilon}(\Theta[\mathcal{N}]) 
&\leq&  \nonumber           \hat{LR}(\Theta[\mathcal{N}^*])
 \\\nonumber  &\leq& \hat{LR}(\mathcal{N}^*)  + \log(1+\delta)  
 \\ &=& \hat{LR}_{\epsilon}(\mathcal{N})  + \log(1+\delta),
\end{eqnarray}  
where the second inequality follows from Lemma~\ref{lemma:LRbound}.
\end{proof}

Now, we formally define the one-shot catalytic dynamic coherence cost of a quantum channel as follows:
\begin{definition}
    For any $\delta,\epsilon \geq 0$ and $d,l\in \mathbb N$, the one-shot catalytic dynamic coherence cost of a quantum channel $\mathcal{N}_A\in \textrm{CPTP}(A_0\to A_1)$ is defined as 
\begin{eqnarray}\label{def:one-shot-cat-cost}
    \tilde{c}_{\epsilon,\delta}^{(1)}(\mathcal{N}_A)= 
      \min\{\log d^2: 
        \Theta_{A'B'\to AB}[\mathcal{F}_{d,A'}\otimes\mathcal{F}_{l,B'}]  =  
        \mathcal{N}'_{A} \otimes \mathcal{F}_{l,B},
          \\\nonumber \Theta_{A'B'\to AB}\in \delta\textrm{-MISC}(A'B'\to AB),
            \frac{1}{2}||\mathcal{N}'_A-\mathcal{N}_A||_{\diamond}\leq \epsilon
    \}.
\end{eqnarray}    
\end{definition}
Before presenting the upper and lower bounds for the one-shot catalytic dynamic coherence cost of a quantum channel, we first introduce the following lemma:
\begin{shaded}
\begin{lemma}\label{lemma:cat-cost}
 For any quantum channel $\mathcal{N}_A\in \textrm{CPTP}(A_0\to A_1)$ and $\epsilon\geq 0$, there exists a quantum channel $\mathcal{M}_{AB}^{\epsilon'}\in \textrm{CPTP}(A_0B_0\to A_1B_1)$ satisfying following properties:
   \begin{equation} \label{eq:M_AB-epsilon}       \mathcal{M}_{AB}^{\epsilon'}=p\mathcal{N}^{\epsilon}_A\otimes\mathcal{F}_{l,B}+(1-p)\mathcal{L}_{AB},
   \end{equation}
\begin{equation}\label{eq:39}
\hat{LR}(\mathcal{M}_{AB}^{\epsilon'})\leq \hat{LR}_{\epsilon'}(\mathcal{N}_A\otimes\mathcal{F}_{l,B}),    
\end{equation}
\begin{equation}\label{eq:40}
p\geq 1-2\epsilon',    
\end{equation}
\begin{equation}\label{eq:41}
\frac{1}{2}||\mathcal{N}^{\epsilon}_A-\mathcal{N}_A||_{\diamond}\leq \epsilon ,
\end{equation}
where $\mathcal{L}_{AB}\in \textrm{CPTP}(A_0B_0\to A_1B_1)$ and $\epsilon'=\frac{\epsilon^2}{2|A_0|^2}$.
\end{lemma}    
\end{shaded}
\begin{proof}
Let $ \tilde{M}_{AB}^{\epsilon'}  \in \textrm{CPTP}(A_0B_0\to A_1B_1)$ satisfies $\hat{LR}( \tilde{\mathcal{M}}_{AB}^{\epsilon'})=
\hat{LR}_{\epsilon'}( \mathcal{N}_A\otimes\mathcal{F}_{l,B} )$, which implies there exists a classical channel $\mathcal{P}_{AB}\in \textrm{C}(A_0B_0\to A_1B_1)$ such that 
\begin{equation}\label{eq:M_ep}
\tilde{M}_{AB}^{\epsilon'}  \leq  2^{\hat{LR}( \tilde{\mathcal{M}}_{AB}^{\epsilon'})}\mathcal{P}_{AB}.
\end{equation}
Now, consider following superchannel $\Omega_{A}\in\hat{\mathfrak O}(A\to A)$ which maps any quantum channel $\mathcal{E}_{A}\in \textrm{CPTP}(A_0\to A_1)$ into following normalized Choi matrix structured form
\begin{eqnarray}\label{eq:misc-lemma}
    \Omega_A[\mathcal{E}_A]  
= \Tr[J^{ \mathcal{F}_{d,A}} J^{\mathcal{E}_{A}}] \mathcal{F}_{d,A}
+     \Tr[(I_{\tilde{A}A}  - J^{ \mathcal{F}_{d,A}}   ) J^{\mathcal{E}_{A}}]  \mathcal{G}_A,
\end{eqnarray}
with $\mathcal{G}_{A}\in \textrm{CPTP}(A_0\to A_1)$ corresponding to the following normalized Choi matrix:
\begin{equation}
   J^{\mathcal{G}_A}=\frac{I-J^{\mathcal{F}_{d,A}} }{d^2-1}.
\end{equation} 
Note that $\Omega_A\in \textrm{MISC}(A\to A)$ due to the fact that for any classical channel $\mathcal{P}_A\in\textrm{C}(A_0\to A_1)$, the following equality chain holds for the normalised Choi matrix $J^{\Omega_A[\mathcal{P}_A]}$:
\begin{eqnarray}
    J^{\Omega_A[\mathcal{P}_A]} &=&\nonumber
      \Tr[J^{ \mathcal{F}_{d,A}} J^{\mathcal{P}_{A}}] J^{\mathcal{F}_{d,A}}
+     \Tr[(I_{\tilde{A}A}  - J^{ \mathcal{F}_{d,A}}   ) J^{\mathcal{P}_{A}}]  J^{\mathcal{G}_A }    
    \\\nonumber &=&  
       \frac{1}{d^2} J^{\mathcal{F}_{d,A}}
+     (1-\frac{1}{d^2}) J^{\mathcal{G}_A }  
    \\ &=&
    \frac{I}{d^2}\in \mathfrak{F},
\end{eqnarray}   
where the second equality follows from $\Tr[J^{ \mathcal{F}_{d,A}} J^{\mathcal{P}_{A}}] = \frac{1}{d^2}$. 

Let $M_{AB}^{\epsilon'}:=\Omega_B[\tilde{M}_{AB}^{\epsilon'}]$. From Eq.~(\ref{eq:misc-lemma}), we find that $\mathcal{M}_{AB}^{\epsilon'}$ has following form
 \begin{equation}  
\mathcal{M}_{AB}^{\epsilon'}=p\mathcal{N}^{\epsilon}_A\otimes\mathcal{F}_{l,B}+(1-p)\mathcal{L}_{AB}.
   \end{equation}
Therefore, we have
\begin{equation}
M_{AB}^{\epsilon'}=\Omega_B[\tilde{M}_{AB}^{\epsilon'}]    \leq 
2^{\hat{LR}( \tilde{\mathcal{M}}_{AB}^{\epsilon'})} \Omega_B [\mathcal{P}_{AB}]=2^{\hat{LR}_{\epsilon'}( \mathcal{N}_A\otimes\mathcal{F}_{l,B} )} \Omega_B [\mathcal{P}_{AB}].
\end{equation}
Which implies that 
\begin{equation}
\hat{LR}(\mathcal{M}_{AB}^{\epsilon'})\leq \hat{LR}_{\epsilon'}(\mathcal{N}_A\otimes\mathcal{F}_{l,B}).  
\end{equation}    
Then, we have 
\begin{eqnarray}
1-\sqrt{F(J^{\Omega_B[\tilde{M}_{AB}^{\epsilon'}]},J^{\mathcal{N}_A\otimes\mathcal{F}_{l,B}})} 
&\leq&\nonumber
    \frac{1}{2} ||J^{\Omega_B[\tilde{M}_{AB}^{\epsilon'}]}-J^{\mathcal{N}_A\otimes\mathcal{F}_{l,B}} ||_1
\\&\leq&\nonumber
   \frac{1}{2} \sup_{\Psi_{ABA'B'}} || \Omega_B[\tilde{M}_{AB}^{\epsilon'}](\Psi_{ABA'B'}) -\Omega_B[\mathcal{N}_A\otimes\mathcal{F}_{l,B}](\Psi_{ABA'B'})||_1
\\&=&\nonumber
\frac{1}{2} || \Omega_B[\tilde{M}_{AB}^{\epsilon'}] - 
     \mathcal{N}_A\otimes\mathcal{F}_{l,B}||_{\diamond}
\\&=& \nonumber
\frac{1}{2} || \Omega_B[\tilde{M}_{AB}^{\epsilon'}] -\Omega_B[\mathcal{N}_A\otimes\mathcal{F}_{l,B}]||_{\diamond}
\\&\leq&\nonumber
\frac{1}{2} || \tilde{M}_{AB}^{\epsilon'} -\mathcal{N}_A\otimes\mathcal{F}_{l,B}||_{\diamond}
\\&\leq&
       \epsilon',
\end{eqnarray}
where the first inequality follows from the Fuchs-van Graaf inequalities:
$1 - \sqrt{F(\rho, \sigma)} \leq \frac{1}{2} \|\rho - \sigma\|_1$, the second inequality follows from the definition of the normalized Choi matrix, the third line follows from the definition of the diamond norm, the fourth line holds because $\Omega_B[\mathcal{F}_{l,B}] = \mathcal{F}_{l,B}$, the fifth line follows from the contractivity of the diamond norm under the superchannel $\Omega_B$, and the last line holds because $\hat{LR}( \tilde{\mathcal{M}}_{AB}^{\epsilon'}) = \hat{LR}_{\epsilon'}( \mathcal{N}_A \otimes \mathcal{F}_{l,B})$.

Thus, we obtain:
\begin{equation}
1-\sqrt{F(J^{\Omega_B[\tilde{M}_{AB}^{\epsilon'}]},J^{\mathcal{N}_A\otimes\mathcal{F}_{l,B}})} 
\leq
\epsilon'
\implies
F(J^{\Omega_B[\tilde{M}_{AB}^{\epsilon'}]},J^{\mathcal{N}_A\otimes\mathcal{F}_{l,B}})
\geq
(1-\epsilon')^2
\geq
1-2\epsilon'.
\end{equation}
Note that 
\begin{eqnarray}
F(J^{\Omega_B[\tilde{M}_{AB}^{\epsilon'}]},                J^{\mathcal{N}_A\otimes\mathcal{F}_{l,B}})
&=&\nonumber
F(pJ^{\mathcal{N}_A^{\epsilon}\otimes\mathcal{F}_{l,B}}+(1-p)J^{\mathcal{L}_{AB}}, 
   J^{\mathcal{N}_A\otimes\mathcal{F}_{l,B}})
\\\nonumber&=&
pF(J^{\mathcal{N}_A^{\epsilon}\otimes\mathcal{F}_{l,B}},   J^{\mathcal{N}_A\otimes\mathcal{F}_{l,B}})
\\&=&
pF(J^{\mathcal{N}_A^{\epsilon}}, J^{\mathcal{N}_A}),
\end{eqnarray}
where the second equality holds because $J^{\mathcal{L}_{AB}}$ and $J^{\mathcal{N}_A\otimes\mathcal{F}_{l,B}}$ have orthogonal images (i.e., $J^{\mathcal{L}_{AB}}J^{\mathcal{N}_A\otimes\mathcal{F}_{l,B}}=0$). From the above, we deduce that $p\geq 1-2\epsilon'$.

Finally, for the upper bound of $\frac{1}{2}||\mathcal{N}^{\epsilon}_A-\mathcal{N}_A||_{\diamond}$, the following inequality chain holds:
\begin{eqnarray}
\frac{1}{2} ||\mathcal{N}^{\epsilon}_A-\mathcal{N}_A||_{\diamond}    
&\leq&\nonumber
|A_0| \frac{1}{2} ||J^{\mathcal{N}^{\epsilon}_A}-J^{\mathcal{N}_A}||_1    
\\\nonumber&\leq&
|A_0| \sqrt{  1- F(J^{\mathcal{N}^{\epsilon}_A}, J^{\mathcal{N}_A})   }
\\\nonumber &\leq&
|A_0| \sqrt{2\epsilon'}
\\ &=&
\epsilon,
\end{eqnarray}
where the first inequality holds follows from $||\mathcal{N}^{\epsilon}_A-\mathcal{N}_A||_{\diamond}     \leq  |A_0|\cdot ||J^{\mathcal{N}^{\epsilon}_A}-J^{\mathcal{N}_A}||_1$~\cite{KIM-one-shot-2021}, the second inequality holds due to the Fuchs-van Graaf inequalities: 
$\frac{1}{2}||\rho-\sigma||_1   \leq  \sqrt{1-F(\rho, \sigma)}$, which completes the proof.
\end{proof}

We now present the upper and lower bounds for the one-shot catalytic dynamic coherence cost of a quantum channel as follows:
\begin{shaded}
\begin{theorem}
%
For any $\delta>0, \epsilon\geq0$, there exists $l\in \mathbb{N}$ with $l^2\geq 1+\frac{1}{\delta}$ such that
\begin{eqnarray}
  \hat{LR}_{\epsilon'} ( \mathcal{N}_A\otimes \mathcal{F}_{l,B}) -\log(l^2(1-2\epsilon'))  +2                    &\geq&
  \tilde{c}_{\epsilon,\delta}^{(1)}(\mathcal{N}_A)    
  \\&\geq&    
  \hat{LR}_{\epsilon}(\mathcal{N}_A\otimes \mathcal{F}_{l,B}) -\log [l^2(1+\delta)  ],
\end{eqnarray}
where $ \epsilon'=\frac{\epsilon^2}{  2|A_0| }   $.
\end{theorem}    
\end{shaded}
\begin{proof}
For the lower bound, consider a superchannel $\Theta_{A\to B}\in\delta\textrm{-MISC}(A\to B)$ and a catalyst $\mathcal{F}_{l,B}$ such that
\begin{equation}
\Theta_{A\to B}[\mathcal{F}_{d,A}\otimes\mathcal{F}_{l,B}]  =  
        \mathcal{N}'_{A} \otimes \mathcal{F}_{l,B}, 
\end{equation}
where $\tilde{c}_{\epsilon,\delta}^{(1)}(\mathcal{N}_A)=\log d^2$ and $ \frac{1}{2}||\mathcal{N}'_A-\mathcal{N}_A||_{\diamond}\leq \epsilon$. Then, we have 

 \begin{eqnarray}
 \hat{LR}_{\epsilon}(\mathcal{N}_A\otimes \mathcal{F}_{l,B}) 
 \nonumber &\leq&             \hat{LR}(\mathcal{N}_A'\otimes\mathcal{F}_{l,B})   
 \\\nonumber &=&  \hat{LR}(\Theta_{A\to B}[\mathcal{F}_{d,A}\otimes\mathcal{F}_{l,B}]) 
 \\\nonumber &\leq&   \hat{LR}(\mathcal{F}_{d,A}\otimes\mathcal{F}_{l,B}) + \log(1+\delta) 
 \\\nonumber          &=&     \log d^2 +  \log l^2 + \log(1+\delta) 
  \\\nonumber         &=&     \tilde{c}_{\epsilon,\delta}^{(1)}(\mathcal{N}_A)  +  \log l^2 + \log(1+\delta)  
   \\        &=&     \tilde{c}_{\epsilon,\delta}^{(1)}(\mathcal{N}_A)  +  \log [l^2(1+\delta)  ], 
 \end{eqnarray}
where the second inequality follow from Lemma~\ref{lemma:LRbound}, the second equality holds because the log-robustness of coherence of a channel is additive under tensor products~\cite{Saxena20}. This completes the proof of the lower bound.

For the upper bound, suppose $\mathcal{M}^{\epsilon'}_{AB}\in \textrm{CPTP}(A_0B_0\to A_1B_1)$ satisfies the conditions given in Eq.(\ref{eq:M_AB-epsilon})-Eq.(\ref{eq:41}). Combining Eq.(\ref{eq:M_AB-epsilon}) with Eq.(\ref{eq:39}), we have 
\begin{eqnarray}
\mathcal{M}^{\epsilon'}_{AB}
&=&\nonumber 
p\mathcal{N}_A^{\epsilon}\otimes\mathcal{F}_{l,B} + (1-p) \mathcal{L}_{AB}
\\&\leq&
2^{ \hat{LR}_{\epsilon'}(\mathcal{N}_A^{\epsilon}\otimes\mathcal{F}_{l,B})}\mathcal{P}_{AB},
\end{eqnarray}
where $\mathcal{P}_{AB}\in\textrm{C}(A_0B_0\to A_1B_1)$. This implies that 
\begin{eqnarray}
\mathcal{N}_A^{\epsilon}\otimes\mathcal{F}_{l,B}
&\leq&\nonumber 
2^{ \hat{LR}_{\epsilon'}(\mathcal{N}_A^{\epsilon}\otimes\mathcal{F}_{l,B})-\log p}\mathcal{P}_{AB} 
\\&\leq&
2^{ \hat{LR}_{\epsilon'}(\mathcal{N}_A^{\epsilon}\otimes\mathcal{F}_{l,B})-\log (1-2\epsilon')} \mathcal{P}_{AB} ,
\end{eqnarray}
where the first inequality holds because $\mathcal{L}_{AB}\geq0$, and the second inequality follows from Eq.(\ref{eq:40}). Therefore, let $s=2^{ \hat{LR}_{\epsilon'}(\mathcal{N}_A^{\epsilon}\otimes\mathcal{F}_{l,B})-\log (1-2\epsilon')}-1$,
we can find a quantum channel $\mathcal{G}_{AB}$ such that 
\begin{equation}\label{eq:99}
    \frac{\mathcal{N}_A^{\epsilon}\otimes\mathcal{F}_{l,B}+s\mathcal{G}_{AB}}{1+s}\in \textrm{C}(A_0B_0\to A_1B_1).
\end{equation}
For any quantum channel $\mathcal{E}_{A'B'}$, consider a superchannel $\Theta_{A'B'\to AB}$ defined as follows:
\begin{equation}
\Theta_{A'B'\to AB}[\mathcal{E}_{A'B'}]  
=
\Tr( J^{ \mathcal{F}_{d,A'}\otimes\mathcal{F}_{l,B'}}  J^{\mathcal{E}_{A'B'}}) \mathcal{N}_A^\epsilon \otimes  \mathcal{F}_{l,B}  
+ \Tr[(I_{ABA'B'}-J^{\mathcal{F}_{d,A'}\otimes\mathcal{F}_{l,B'}}) J^{\mathcal{E}_{A'B'}} ] \mathcal{G}_{AB}.
\end{equation}
We will demonstrate that the superchannel $\Theta_{A'B'\to AB}$ above is a feasible solution to Eq.(\ref{def:one-shot-cat-cost}). 
First, it is straightforward to verify that
\begin{equation}
\Theta_{A'B'\to AB}[\mathcal{F}_{d,A'}\otimes\mathcal{F}_{l,B'}]  =  
        \mathcal{N}^{\epsilon}_{A} \otimes \mathcal{F}_{l,B}    
\end{equation}
Second, we need to show that $\Theta_{A'B'\to AB}\in \delta$-MISC$(A'B'\to AB)$. For any classical channel $\mathcal{P}_{A'B'}$, we have that
\begin{eqnarray}
\Theta_{A'B'\to AB}[\mathcal{P}_{A'B'}]  
&=&\nonumber
\Tr( J^{ \mathcal{F}_{d,A'}\otimes\mathcal{F}_{l,B'}}  J^{\mathcal{P}_{A'B'}}) \mathcal{N}_A^\epsilon \otimes  \mathcal{F}_{l,B}  
+ \Tr[(I_{ABA'B'}-J^{\mathcal{F}_{d,A'}\otimes\mathcal{F}_{l,B'}}) J^{\mathcal{P}_{A'B'}} ] \mathcal{G}_{AB}
\\ &=&
q\frac{\mathcal{N}_A^{\epsilon}\otimes\mathcal{F}_{l,B}+s\mathcal{G}_{AB}}{1+s}+(1-q)\mathcal{G}_{AB},
\end{eqnarray}
where $ q=(1+s)\Tr( J^{ \mathcal{F}_{d,A'}\otimes\mathcal{F}_{l,B'}} J^{\mathcal{P}_{A'B'}})=\frac{1+s}{d^2l^2}$. To ensure $0\leq q\leq1$, we assume that $d=\lceil \frac{\sqrt{1+s}}{l}\rceil$. Thus, we can bound the robustness of coherence of $\Theta_{A'B'\to AB}[\mathcal{P}_{A'B'}] $ as follows
\begin{eqnarray}
 \hat{C}_R(\Theta_{A'B'\to AB}[\mathcal{P}_{A'B'}])
  &=&\nonumber
   \hat{C}_R(q\frac{\mathcal{N}_A^{\epsilon}\otimes\mathcal{F}_{l,B}+s\mathcal{G}_{AB}}{1+s}+(1-q)\mathcal{G}_{AB})    
     \\\nonumber&\leq&   
       q \hat{C}_R(\frac{\mathcal{N}_A^{\epsilon}\otimes\mathcal{F}_{l,B}+s\mathcal{G}_{AB}}{1+s}) + (1-q)\hat{C}_R(\mathcal{G}_{AB})  
       \\\nonumber &\leq&
        \hat{C}_R(\mathcal{G}_{AB})    
        \\\nonumber &\leq&
         \frac{1}{\hat{C}_R(\mathcal{N}_A^{\epsilon}\otimes\mathcal{F}_{l,B}) }
         \\\nonumber &\leq&
           \frac{1}{\hat{C}_R(\mathcal{F}_{l,B}) }
         \\   &=&
         \frac{1}{l^2-1},
\end{eqnarray}
where the first inequality follows from the convexity of the robustness of coherence, the second inequality follows from Eq.(\ref{eq:99}), the third inequality holds is due to  
\begin{equation}
\frac{\mathcal{N}_A^{\epsilon}\otimes\mathcal{F}_{l,B}+s\mathcal{G}_{AB}}{1+s}
=
\frac{s^{-1}\mathcal{N}_A^{\epsilon}\otimes\mathcal{F}_{l,B}+\mathcal{G}_{AB}}{1+s^{-1}}
\in \textrm{C} (A_0B_0\to A_1B_1) ,   
\end{equation}
the fourth inequality holds follows from the monotonicity of robustness under discarding a system. 

Thus, to easure $\hat{C}_R(\Theta_{A'B'\to AB}[\mathcal{P}_{A'B'}])\leq\frac{1}{l^2+1} \leq\delta$, we assume $l\geq \sqrt{1+\frac{1}{\delta}}$. Consequently, we set 
\begin{equation}
    d=\lceil\frac{\sqrt{1+s}}{l}\rceil=\lceil \frac{\sqrt{2^{ \hat{LR}_{\epsilon'}(\mathcal{N}_A^{\epsilon}\otimes\mathcal{F}_{l,B})-\log (1-2\epsilon')} }}{l}  \rceil,
\end{equation}
and $l\in\mathbb{N}$.
Thus, we obtain
\begin{eqnarray}
  \hat{LR}_{\epsilon'} ( \mathcal{N}_A\otimes \mathcal{F}_{l,B}) -\log(l^2(1-2\epsilon'))  +2             &\geq&
        \log d^2  
         \\&\geq&    
   \tilde{c}_{\epsilon,\delta}^{(1)}(\mathcal{N}_A) ,
\end{eqnarray}
where $ \epsilon'=\frac{\epsilon^2}{  2|A_0| }   $.
\end{proof}

\section{Conclusion and Discussion}\label{sec:conclusion}
In this paper, we regard the classical channel as the free channel and the QFT channel as the golden unit of dynamic coherence resources. The role of the QFT channel is analogous to that of the maximally coherent state within the framework of static coherence resources. We investigated the one-shot manipulation of dynamic coherence under two kinds of free superchannels: the one-shot dynamic coherence cost and one-shot dynamic coherence distillation. Specifically, we found that the one-shot dynamic coherence cost under MISC is bounded by the log-robustness of coherence, while the one-shot dynamic coherence cost under DISCs is constrained by the dephasing log-robustness of coherence. Additionally, the one-shot dynamic coherence distillation under MISCs is limited by the hypothesis-testing relative entropy of coherence, whereas the one-shot dynamic coherence distillation under DISCs is bounded by the dephasing hypothesis-testing relative entropy of coherence. Finally, we examine the catalytic scenario in which additional dynamic coherence is supplied and subsequently returned after the superchannels. We found that the one-shot dynamic coherence cost under $\delta$-MISCs is bounded by the log-robustness of coherence. The results presented above provide an operational interpretation of the log-robustness of coherence, dephasing log-robustness of coherence, hypothesis-testing relative entropy of coherence, and dephasing hypothesis-testing relative entropy of coherence. These results provide constraints on the ability of QFT channels to simulate general channels, as well as on the ability of general channels to simulate QFT channels.

We note that an alternative formulation of the one-shot manipulation of dynamic coherence under free superchannels recently appeared in Ref.\cite{Saxena20}. The authors of Ref.\cite{Saxena20} studied the exact one-shot dynamic coherence cost and one-shot dynamic coherence distillation. The definitions of these concepts depend on static states rather than on the channel itself. For instance, the one-shot dynamic coherence distillation is defined as the logarithm of the dimension of the maximally coherent state $\psi^+_{B_1}$ that can be transformed from a given channel $\mathcal{N}_A$ using a free superchannel $\Theta_{A\to B}$. That is
\begin{equation}
    \tilde{d}^{(1)}_{\epsilon, \hat{\mathfrak{O}}} (\mathcal{N}_A)= \max \{ \log |B_1|: F(\Theta_{A\to B} [\mathcal{N}_A], \psi^+_{B_1})\geq 1-\epsilon, \Theta_{A\to B}\in\hat{\mathfrak{O}}\}.
\end{equation}
Since this study focuses on dynamic coherence theory, the concepts defined using the QFT channel, such as dynamic coherence distillation (e.g., Eq.(\ref{def:dynamic-coh-dis})), are more aligned with the terminology of dynamic coherence resource theory~\cite{Gour2019,Takagi2022}. This type of research also appears in the study of dynamic entanglement resource theory. For example, in Ref.~\cite{KIM-one-shot-2021}, the authors considered the $K$-swap channel to be the golden unit within this dynamic resource theory. Meanwhile, viewing the QFT channel as the golden unit also provides theoretical insights for understanding its simulation or utilization~\cite{Jones2024hadamard}. Finally, we believe that our study can also provide insights for other dynamic QRTs, such as "dynamic imaginarity"\cite{Hickey2018,Wu2021b,Wu2021a,Xu2023-PhysRevA.108.062203} and "dynamic superposition"\cite{Theurer2017-PhysRevLett.119.230401,Huseyin2022-PhysRevA.105.042410,Torun2021-PhysRevA.103.032416}.

\section*{acknowledgments}
I thank Ping Li and Mingfei Ye for useful discussions on an early manuscript. I also express my gratitude to the anonymous referees, whose valuable feedback has greatly contributed to improving the paper. This work is supported by the National Natural Science Foundation of China (Grant No. 62001274, Grant No. 12071271 and Grant No. 62171266).
\sloppy
\bibliographystyle{abbrv}
\bibliography{main}

\fussy
\appendix
\section{The golden units in the dynamic QRT of coherence}\label{app:golden-units}
In this section, we demonstrate that \textit{QFT channels} serve as \textit{the golden units} in the dynamic resource theory of coherence where the free superchannels are \textrm{MISC}. Specifically, we show that any other channel can be obtained from a QFT channel through \textrm{MISC}.

\begin{shaded}
\begin{proposition}
Let \( |A_0| = |A_1| = |B_0| = |B_1| = d \) and \( \mathcal{F}_d \in \textrm{CPTP}(A_0\to A_1) \) be a QFT channel. For any quantum channel \( \mathcal{N} \in \textrm{CPTP}(B_0\to B_1) \), there exists \( \Theta_{A\to B} \in \textrm{MISC}(A\to B) \) such that \( \Theta_{A\to B}[\mathcal{F}_d] = \mathcal{N}\). 
\end{proposition}    
\end{shaded}
\begin{proof}
Since a superchannel acting on a QFT channel can be expressed as  
\begin{equation}
\Theta_{A \to B}[\mathcal{F}_d] = \mathcal{Q}_{A_1E \to B_1} \circ (\text{id}_E\otimes\mathcal{F}_d ) \circ \mathcal{P}_{B_0 \to A_0E},
\end{equation}
where \( \mathcal{P}_{B_0 \to A_0E} \) and \( \mathcal{Q}_{A_1E \to B_1} \) are the pre-processing and post-processing channels, respectively, we can construct a superchannel \( \Theta_{A\to B} \in \textrm{MISC}(A\to B) \) by choosing \( \mathcal{P}_{B_0 \to A_0E} \) as a DI channel and \( \mathcal{Q}_{A_1E \to B_1} \) as an MIO channel~\cite{Saxena20}. 
We define the quantum channels \( \mathcal{P}_{B_0 \to A_0E} \) and \( \mathcal{Q}_{A_1E \to B_1} \) as follows. The channel \( \mathcal{P}_{B_0 \to A_0E} \) acts as a DI channel and is given by  
\(\mathcal{P}_{B_0 \to A_0E}(\cdot) = \cdot \otimes \ketbra{0}{0}_{A_0}\).
Similarly, we define \( \mathcal{Q}_{A_1E \to B_1} \) as an MIO channel, which is specified by  
\( \mathcal{Q}_{A_1E \to B_1}(\cdot \otimes \psi^+_{A_1}) = \mathcal{N}(\cdot) \) .
The reason \( \mathcal{Q}_{A_1E \to B_1} \) qualifies as an MIO channel follows from Theorem 6 in Ref.~\cite{diaz2018-using}, which states that \( \mathcal{Q}_{A_1E \to B_1} \) is an MIO channel if and only if  
\( LR(\psi^+_{A_1}) \geq \max_{i} LR(\mathcal{N}(\ketbra{i}{i})) \) where \(\{\ket{i}\}\) is the fixed basis.
Then, we have 
\begin{eqnarray}
\Theta_{A \to B}[\mathcal{F}_d](\cdot) &=&
\mathcal{Q}_{A_1E \to B_1} \circ (\text{id}_E\otimes\mathcal{F}_d ) \circ \mathcal{P}_{B_0 \to A_0E}(\cdot)
\\&=&
\mathcal{Q}_{A_1E \to B_1} \circ (\text{id}_E\otimes\mathcal{F}_d ) (\cdot\otimes\ketbra{0}{0}_{A_0})
\\&=&
\mathcal{Q}_{A_1E \to B_1} (\cdot\otimes\psi^+_{A_1})
\\&=&
\mathcal{N}(\cdot).
\end{eqnarray}
This completes the proof.
\end{proof}

\section{The maximal replacement channels can not be the golden unit in the dynamic QRT of coherence}\label{app:maximal-replacement}
Let \(\mathcal{F}_d\in\textrm{CPTP}(A_0\to A_1)\) be a QFT channel and \(\mathcal{R}_d(\cdot)\in\textrm{CPTP}(A_0\to A_1)\) be a maximal replacement channel with \(|A_0| = |A_1| = d\),  . We will show that 
\begin{equation}
    \hat{LR}(\mathcal{F}_d)=\hat{LR}_{\Delta}(\mathcal{F}_d)=\log d^2,
\end{equation}
and
\begin{equation}
    \hat{LR}(\mathcal{R}_d)=\hat{LR}_{\Delta}(\mathcal{R}_d)=\log d.
\end{equation}
The equations \(\hat{LR}(\mathcal{F}_d) = \log d^2\) and \(\hat{LR}(\mathcal{R}_d) = \log d\) were established in Ref.~\cite{Saxena20}. From the definition of dephasing log-robustness, we have the inequality \( \log d^2 \geq \hat{LR}_{\Delta}(\mathcal{F}_d) \geq \hat{LR}(\mathcal{F}_d) = \log d^2 \), which directly implies that \(\hat{LR}_{\Delta}(\mathcal{F}_d) = \log d^2\).
For the maximal replacement channel \(\mathcal{R}_d\), we have 
\begin{eqnarray}
\hat{LR}_{\Delta}(\mathcal{R}_d) &=& \hat{D}_{\max} (\mathcal{R}_d||\Delta[\mathcal{R}_d])
\\&=&
D_{\max} (J^{\mathcal{R}_d}||J^{\Delta[\mathcal{R}_d]})
\\&=&
D_{\max} (\frac{I}{d}\otimes\psi_d^+||\frac{I}{d}\otimes\frac{I}{d})
\\&=&
\log d.
\end{eqnarray}
In conclusion, we observe that the log-robustness (or dephasing log-robustness) of coherence for the QFT channel is consistently twice that of the maximal coherent replacement channel. This implies that under \textrm{MISCs} (\textrm{DISCs}), two maximal replacement channels are required to simulate a QFT channel.
\section{The maximal replacement channels \(\mathcal{R}_d\) can be converted from the QRT channels \(\mathcal{F}_d\) via \( \Theta\in\textrm{DISC} \)}\label{app:QRT-convert-Replacement}
Let \(\Theta_{A \to B}[\cdot] = \mathcal{Q}_{A_1E \to B_1} \circ (\text{id}_E\otimes\cdot ) \circ \mathcal{P}_{B_0 \to A_0E}\) be a superchannel and \( |A_0| = |A_1| = |B_0| = |B_1| = d \). To ensure that \(\Theta_{A \to B}\) belongs to \(\textrm{DISC}\), we set \(\mathcal{P}_{B_0 \to A_0E}(\cdot) = \sigma_E \otimes \ketbra{0}{0}\), where \(\sigma_E \in \mathfrak{F}\), and define \( \mathcal{Q}_{A_1E \to B_1}(\cdot) = \Tr_E(\cdot) \) as the partial trace operation. It is straightforward to verify that both \(\mathcal{P}_{B_0 \to A_0E}\) and \(\mathcal{Q}_{A_1E \to B_1} \) are DIO channels, ensuring that \(\Theta_{A \to B} \in \textrm{DISC}(A\to B)\). A simple verification confirms that \(\Theta_{A \to B}[\mathcal{F}_d](\cdot) = \mathcal{R}_d(\cdot)=\Tr(\cdot)\psi_d^+\).





\renewcommand{\theequation}{S\arabic{equation}}
\renewcommand{\thesubsection}{\normalsize{Supplementary Note \arabic{subsection}}}
\renewcommand{\theproposition}{S\arabic{proposition}}
\renewcommand{\thedefinition}{S\arabic{definition}}
\renewcommand{\thefigure}{S\arabic{figure}}
\setcounter{equation}{0}
\setcounter{table}{0}
\setcounter{section}{0}
\setcounter{proposition}{0}
\setcounter{definition}{0}
\setcounter{figure}{0}

%



\end{document}